\documentclass{article}
\usepackage{todonotes}
\usepackage{geometry}
\usepackage{amssymb,bm}
\usepackage{url}
\usepackage{booktabs}
\usepackage{algorithm}
\usepackage{algorithmic}
\usepackage[numbers,sort&compress]{natbib}
\graphicspath{{Figs/}}
\usepackage[caption=false]{subfig}

\usepackage[shortlabels]{enumitem}
\usepackage{nicefrac}       
\usepackage{graphicx}

\usepackage{amsthm,amsmath}
\usepackage{float}
\usepackage{cleveref}

\newtheorem{theorem}{Theorem}
\newtheorem{lemma}{Lemma}
\newtheorem{corollary}{Corollary}

\theoremstyle{definition}
\newtheorem{definition}{Definition}

\theoremstyle{remark}
\newtheorem{remark}{Remark}

\crefname{subsection}{Section}{Sections} 
\Crefname{subsection}{Section}{Sections}

\DeclareMathOperator{\rank}{rank}
\DeclareMathOperator{\Tr}{Tr}
\DeclareMathOperator{\Var}{Var}
\DeclareMathOperator{\arsinh}{arcsinh}
\DeclareMathOperator*{\argmin}{arg\,min}

\newcommand{\cO}{\mathcal{O}}
\newcommand{\Ee}{\mathbb E}
\newcommand{\Prr}{\mathbb P}
\newcommand{\norm}[1]{\left\lVert #1\right\rVert}
\newcommand{\abs}[1]{\left\lvert #1\right\rvert}
\newcommand{\ip}[2]{\left\langle #1,#2\right\rangle}
\newcommand{\F}{\mathrm F}
\newcommand{\op}{\mathrm{op}}
\newcommand{\D}{\mathcal D}
\newcommand{\A}{\mathcal A}

\newcommand{\brho}{\bm{\rho}}
\newcommand{\brhostar}{\bm{\rho}_\star}
\newcommand{\bP}{\bm{P}}
\newcommand{\bX}{\bm{X}}
\newcommand{\bZ}{\bm{Z}}
\newcommand{\bW}{\bm{W}}
\newcommand{\bH}{\bm{H}}
\newcommand{\bG}{\bm{G}}
\newcommand{\bU}{\bm{U}}
\newcommand{\bI}{\bm{I}}
\newcommand{\bDelta}{\bm{\Delta}}
\newcommand{\bLambda}{\bm{\Lambda}}

\newcommand{\bu}{\bm{u}}

\newcommand{\bmu}{\bm{\mu}}
\newcommand{\bz}{\bm{z}}
\newcommand{\bone}{\bm{1}}
\newcommand{\Ical}{\mathcal I}
\title{Quantized Low-Rank Quantum State Tomography: Hyperbolic Quantization and \\ Riemannian Least-Squares Recovery}
\author{
HanQin Cai\thanks{School of Data, Mathematical, and Statistical Sciences and Department of Computer Science,
University of Central Florida,
Orlando, FL 32816, USA (\texttt{hqcai@ucf.edu}).}
\and Longxiu Huang\thanks{Department of Computational Mathematics, Science and Engineering, Department of Mathematics, Michigan State University, East Lansing, MI 48824, USA (\texttt{huangl3@msu.edu}).}
\and Juntao You\thanks{(Corresponding author.) School of Artificial Intelligence, and Hubei Center for Applied Mathematics, Wuhan University, Wuhan, China (\texttt{youjuntao@whu.edu.cn}).}
}
\date{}

\begin{document}
\maketitle

\begin{abstract}
We study low-rank quantum state tomography from finite-bit Pauli batch responses. To avoid bias introduced by generic quantization, we propose HyperQuant, a mean-preserving hyperbolic quantizer adapted to the second-moment scale of Pauli responses. We establish minimax distortion guarantees and show that exact mean preservation enables direct rank-constrained least-squares recovery without altering the population target. We derive nonasymptotic recovery guarantees and an explicit bit--shot tradeoff under which finite-bit responses retain the error order of unquantized batch averages using fewer response bits. For efficient computation, we develop QuantRGD, a Riemannian gradient method with provable linear convergence to the corresponding statistical neighborhood under explicit resource conditions. Numerical experiments validate the predicted quantization, recovery, and convergence behavior.
\end{abstract}


\section{Introduction}\label{sec:intro}
Quantum state tomography (QST) aims to recover an unknown quantum state from measurements performed on independently prepared quantum systems. For an $n$-qubit system, the density matrix has dimension $d=2^n$, and a general quantum state contains $d^2-1$ real degrees of freedom. When the state has rank at most $r\ll d$, low-rank QST can substantially reduce the number of measurements required for accurate recovery. In particular, a rank-$r$ density matrix can be recovered from nearly $rd$ randomly sampled Pauli observables, up to logarithmic factors \citep{gross2010qst,liu2011universal,liu2012compressed}. Finite-shot recovery guarantees and projected least-squares estimators have also been developed for repeated Pauli measurements \citep{xia2017pauli,guta2020fast}.

Most existing QST models assume that the classical measurement responses are available to the recovery algorithm with sufficiently high precision. In practice, however, quantum measurements are followed by classical storage, communication, and processing, where the precision and throughput of measurement records are finite resources \citep{reilly2015interface,danjou2021readout,google2025qec,prathapan2022cryo,carreravazquez2024realtime}. In this paper, we study low-rank QST under a finite-bit representation of each Pauli batch response. For every sampled Pauli observable, we perform $\ell$ repeated measurements and form their empirical average, which is then encoded using $b$ bits. This separates three resources: the number of Pauli settings $M$, the total number of quantum copies $N=M\ell$, and the response-bit budget $B=Mb$. The central question is \textit{whether finite-bit responses can retain the statistical accuracy of the unquantized batch averages while still admitting efficient recovery.}

A key difficulty is that generic quantization can introduce bias into the response, thereby distorting the underlying Pauli measurement information. To avoid this issue, we impose exact mean preservation and propose \textit{HyperQuant}, a hyperbolic quantizer adapted to the second-moment scale of Pauli responses. Mean preservation ensures that the finite-bit response remains unbiased for the underlying Pauli coefficient, while the nonuniform hyperbolic alphabet controls the additional quantization variance. We establish minimax guarantees under a weighted conditional variance criterion and further characterize the minimax dependence on the alphabet size and dimension in the high-shot pure-state regime.

An important consequence of mean preservation is that the finite-bit responses can be used directly in a rank-constrained least-squares formulation without altering its population target. We establish nonasymptotic recovery guarantees that explicitly characterize the effects of finite-shot noise and quantization, and derive a bit--shot tradeoff under which the finite-bit estimator retains the error order of the unquantized batch averages. For efficient computation, we further develop \textit{QuantRGD}, a Riemannian gradient method operating directly on the finite-bit responses, and establish its linear convergence to a statistical neighborhood under explicit resource conditions.

The finite-bit formulation is related to distributed statistical estimation, where communication and sampling are treated as distinct resources \citep{zhang2013communication,han2018geometric,barnes2020fisher,suresh2017mean}, communication-limited quantum inference \citep{doosti2026distributed,mattig2026distributed}, and quantized low-rank recovery under one-bit or multilevel observation models \citep{cai2013maxnorm,davenport2014onebit,ni2016optimal,klopp2015adaptive,bhaskar2016probabilistic,shen2019robust,ding2020nonquadratic}. Related approaches include noise-shaping quantization \citep{lybrand2019quantization}, dithered matrix completion \citep{chen2023dithered}, and one-bit least-squares recovery \citep{huang2018robust}. Unlike these models, we quantize batch averages of repeated binary Pauli outcomes while preserving their conditional means. 

Factorized and Riemannian algorithms have been analyzed for QST after forming an unquantized or real-valued response vector \citep{kyrillidis2018provable,kim2023fast,hsu2024rgd}. Related Riemannian techniques have also been developed for structured low-rank recovery problems \citep{vandereycken2013low,wei2016guarantees,cai2019accelerated,hamm2022RieCUR,smith2026provableEDMC}. In contrast, QuantRGD operates directly on finite-bit Pauli responses while retaining the rank-constrained least-squares formulation enabled by mean preservation. Its convergence analysis therefore captures the joint effects of finite-shot noise and quantization.

\Cref{tab:intro} summarizes representative low-rank QST methods and their statistical and computational requirements.
\begin{table}[H]
\centering
\caption{Representative low-rank QST results.
Here $N=M\ell$,  $\kappa=\frac{\lambda_1(\rho_\star)}{\lambda_r(\rho_\star)}$, and $\varepsilon\lesssim \lambda_r(\rho_\star) $ denotes the target Frobenius accuracy. $\overline Y_i$ is defined in \eqref{def:Ym}. $Q^{\rm hyp}$ denotes the proposed hyperbolic quantizer.}
\label{tab:intro}
\scriptsize
\setlength{\tabcolsep}{2.5pt}
\renewcommand{\arraystretch}{1.3}
\resizebox{\textwidth}{!}{%
\begin{tabular}{l|cccccc}
\toprule
Method& response& \shortstack{Pauli-setting\\ requirement $M$}& \shortstack{copy\\complexity $N$}& \shortstack{iteration\\complexity}& \shortstack{per-iteration\\complexity}& \shortstack{memory\\complexity} \\
\midrule
ConvexCS (SVT) \citep{liu2012compressed,xia2017pauli,cai2010singular}& $\overline Y_i$& $\widetilde \cO(rd)$& $\widetilde \cO(rd^2/\varepsilon^2)$& sublinear& $\cO(Md+d^3)$& $\cO(d^2)$\\
\hline
ProjFGD \citep{kyrillidis2018provable}& $\Tr(\bP_i\brhostar)$& $\widetilde \cO(\kappa^2r^2d)$& --& $\cO(\kappa\log(1/\varepsilon))$& $\cO(Mdr+dr^2)$& $\cO(dr)$ \\
\hline
MiFGD \citep{kim2023fast}& $\Tr(\bP_i\brhostar)$& $\widetilde \cO(\kappa^2r^2d)$& --& $\cO(\kappa^{\frac{1}{2}}\log(1/\varepsilon))$& $\cO(Mdr+dr^2)$& $\cO(dr)$ \\
\hline
RGD \citep{hsu2024rgd}& $\overline Y_i$& $\widetilde \cO(\kappa^2r^2d)$& $\widetilde \cO(rd^2/\varepsilon^2)$& $\cO(\log(1/\varepsilon))$& $\cO(Mdr+dr^2+r^3)$& $\cO(dr)$ \\
\hline
HyperQuant LS (\Cref{sec:least-squares})& $Q^{\rm hyp}(\overline Y_i)$& $\widetilde \cO(rd)$& $\widetilde \cO(rd^2/\varepsilon^2)$& --& solver dependent& solver dependent \\
\hline
QuantRGD (\Cref{sec:rgd})& $Q^{\rm hyp}(\overline Y_i)$& $\widetilde \cO(\kappa^2r^2d)$& $\widetilde \cO(rd^2/\varepsilon^2)$& $\cO(\log(1/\varepsilon))$& $\cO(Mdr+dr^2+r^3)$& $\cO(dr)$ \\
\bottomrule
\end{tabular}%
}
\end{table}
\subsection{Main contributions}

In this paper, our main contributions are fourfold.

\begin{enumerate}[(i)]
\item We propose \textit{HyperQuant}, a mean-preserving hyperbolic quantizer adapted to the second-moment scale of Pauli responses. We establish minimax guarantees under a weighted conditional variance criterion and show that, in the high-shot pure-state regime, the minimax quantization distortion scales as $\Theta \left(\log^2 d/(dK^2)\right)$. See \Cref{sec:quantization}.

\item We establish nonasymptotic recovery guarantees for rank-constrained least squares directly from the finite-bit responses. Under the bit--shot matching condition, every global minimizer satisfies
$\norm{\widehat{\brho}_{\rm LS}-\brhostar}_{\F}
\lesssim d\sqrt r\max\{\sqrt{\frac{\log d}{N}},\frac{\log d}{M}\}$
retaining the error order of the unquantized batch averages. In particular, when $\ell\leq d$, only $\cO(\log\log(e+\ell))$ bits per Pauli batch are sufficient. See \Cref{sec:least-squares}.

\item We develop \textit{QuantRGD}, a Riemannian gradient method operating directly on the finite-bit responses. Each iteration costs $\cO(Mdr+dr^2+r^3)$ flops with $\cO(dr)$ working memory. We establish recovery guarantees for QuantRGD: under explicit resource conditions, its spectral initializer enters the contraction region, after which the iterates converge linearly to the statistical neighborhood determined jointly by finite-shot noise and quantization. See \Cref{sec:rgd}.

\item We provide extensive numerical experiments validating the theoretical findings. The results demonstrate the distortion advantage of HyperQuant over uniform quantization, the effectiveness of finite-bit low-rank recovery under constrained response budgets, and the convergence performance of QuantRGD relative to representative unquantized methods. See \Cref{sec:experiments}.
\end{enumerate}

\paragraph{Notation}
Bold symbols denote matrices and vectors, with their roles and dimensions clear from context, while nonbold letters denote scalars. For a matrix $\bX$, $\bX^\dagger$, $\rank(\bX)$, and $\Tr(\bX)$ denote its conjugate transpose, rank, and trace, respectively. We use $\norm{\bX}_{\op}$, $\norm{\bX}_{\F}$, and $\norm{\bX}_*$ for its operator, Frobenius, and nuclear norms. The $i$-th singular value of $\bX$ is denoted by $\sigma_i(\bX)$ and  its eigenvalues are arranged as $\lambda_1(\bX)\ge\cdots\ge\lambda_d(\bX)$ whenever algebraic order is used; an alternative ordering, such as decreasing magnitude, is stated explicitly. We write $\bX\succeq\bm 0$ and $\bX\preceq\bZ$ for the Loewner order. The Frobenius inner product is $\ip{\bX}{\bZ}=\operatorname{Re}\Tr(\bX^\dagger\bZ)$, which reduces to $\Tr(\bX\bZ)$ for Hermitian matrices. For a Hermitian matrix $\bZ$, we define $\norm{\bZ}_{(s)}=(\sum_{j=1}^{\min\{s,d\}}\sigma_j(\bZ)^2)^\frac{1}{2}$. For a linear operator $\A$, $\A^*$ denotes its adjoint with respect to the Euclidean and Frobenius inner products. We write $\Ee$, $\Prr$, and $\Var$ for expectation, probability, and variance. The notation $X\sim\mu$ means that $X$ has distribution $\mu$, $\mathrm{Unif}(\mathcal S)$ denotes the uniform distribution on a finite set $\mathcal S$, and $\delta_x$ denotes the Dirac point mass at $x$. Unless otherwise stated, $\log$ is the natural logarithm and $\log_2$ is base two. For two nonnegative quantities $f$ and $g$, the relation $f=\cO(g)$ means that $f/g$ is bounded by a positive constant independent of the problem parameters, while $f=\Theta(g)$ means that both $f=\cO(g)$ and $g=\cO(f)$ hold.  We write $f\lesssim g$, $f\gtrsim g$, and $f\asymp g$ as shorthand for $f=\cO(g)$, $g=\cO(f)$, and $f=\Theta(g)$, respectively. The notation $\widetilde \cO(\cdot)$ suppresses logarithmic factors.

\section{Observation Model and Hyperbolic Quantization}\label{sec:setup}\label{sec:quantization}

This section formulates the finite-bit Pauli observation model and develops a mean-preserving \textbf{hyper}bolic \textbf{quant}izer, termed \textit{HyperQuant}. Its alphabet is adapted to the second-moment scale of Pauli responses to control distortion under a finite-bit budget. We first describe the observation model and then present the construction and guarantees of HyperQuant.

\subsection{Problem setup}
Let
\[\D_r=\left\{\brho\in\mathbb C^{d\times d}:\brho=\brho^\dagger,\ \brho\succeq0,\ \Tr(\brho)=1,\ \rank(\brho)\le r\right\}\]
denote the set of density matrices with rank at most $r$. Our goal is to recover an unknown quantum state $\brhostar\in\D_r$ from finite-bit Pauli responses.

Specifically, for each $i=1,\ldots,M$, we draw $\bP_i$ independently and uniformly at random from the $d^2$ Pauli observables. At the selected setting $\bP_i$, we measure $\ell$ fresh copies of $\brhostar$ and observe
\[Y_{ij}\in\{-1,1\},\qquad\Prr(Y_{ij}=1\mid \bP_i)=\frac{1+\Tr(\bP_i\brhostar)}2,\qquad j=1,\ldots,\ell.\]
We then retain the $b$-bit quantized response
\begin{equation}\label{def:Ym}
\widetilde Y_i = Q(\overline Y_i),\qquad \text{where}\quad \overline Y_i =\frac1\ell\sum_{j=1}^{\ell}Y_{ij},\quad  i=1,\ldots,M,
\end{equation}
and the quantizer $Q$ will be specified below. Thus, $M$ denotes the number of random Pauli settings, $\ell$ is the number of copies measured at each setting, and $b$ is the number of bits used to encode each quantized response. The total numbers of quantum copies and response bits are $N=M\ell$ and $B=Mb$, respectively. Here $B$ counts only the fixed-length response messages. The Pauli settings are assumed to be known to the decoder, for example through a shared random seed. The statistical problem is therefore to recover $\brhostar$ from the quantized responses $\{(\bP_i,\widetilde Y_i)\}_{i=1}^M$. Unlike standard QST, the decoder no longer has access to the batch averages themselves. Quantization introduces an additional source of distortion, while the numbers of Pauli settings, quantum copies, and transmitted bits constrain the recovery problem. A desired procedure should retain the information in the measurements while maintaining sample efficiency and computational tractability. This leads to our central question:

\noindent
\emph{Under resource budgets $(M,N,B)$, can quantized Pauli responses retain the statistical accuracy of their unquantized counterparts while admitting provably efficient recovery?}

We answer this question by designing a mean-preserving quantizer, characterizing its minimax distortion and bit--shot tradeoff, and establishing statistical and computational recovery guarantees.

\subsection{HyperQuant and minimax guarantees}
The batch average $\overline Y_i$ is an unbiased estimator of the Pauli coefficient $\Tr(\bP_i\brhostar)$. A generic finite-alphabet representation, however, can introduce bias and thereby distort the underlying Pauli coefficient. We therefore impose exact mean preservation and then choose the alphabet to control the additional quantization variance.

\begin{definition}\label{def:mean-preserving}
Let $K=2^b$. A $K$-level scalar quantizer has an ordered alphabet $-1=q_0<q_1<\cdots<q_{K-1}=1$. It is \emph{mean-preserving} if
\[\Ee[Q(x)\mid x]=x, \qquad \forall x\in[-1,1].\]
For a fixed alphabet and $x\in[q_j,q_{j+1}]$, the \emph{adjacent stochastic rounding} rule is 
\[Q(x)=
\begin{cases}
q_j,&\text{with probability }\displaystyle\frac{q_{j+1}-x}{q_{j+1}-q_j},\\[1.1ex]
q_{j+1},&\text{with probability }\displaystyle\frac{x-q_j}{q_{j+1}-q_j}.
\end{cases}
\]
This is the unique mean-preserving rule supported on the two levels that bracket $x$.   For uniformly spaced levels \(q_i=-1+\frac{2i}{K-1},\) we call $Q$ the uniform quantizer.
\end{definition}

For $x\in[q_j,q_{j+1}]$, adjacent stochastic rounding satisfies
\begin{equation}\label{eq:adjacent-variance}\Var(Q(x)\mid x)=(x-q_j)(q_{j+1}-x). \end{equation}
Thus, once mean preservation determines the rounding probabilities, the remaining design freedom lies in the placement of the alphabet levels. The relevant scale is determined by the second moment of the quantity being quantized. For the Pauli batch response, Pauli orthogonality together with the conditional second moment of the binary outcomes gives
\begin{equation}\label{eq:pauli-batch-second-moment}
\Ee\overline Y_i^2=\frac{1}{\ell}+\left(1-\frac{1}{\ell}\right)\frac{\Tr(\brhostar^2)}{d}\le \frac1\ell+\frac{1-1/\ell}{d}.
\end{equation}
Accordingly, we define
\begin{equation}\label{eq:pauli-second-moment-scale}
\nu_{\ell,d}=\frac1\ell+\frac{1-1/\ell}{d},\qquad A_{\ell,d}=\arsinh\!\left(\nu_{\ell,d}^{-1/2}\right).
\end{equation}
Eq.~\eqref{eq:adjacent-variance} suggests balancing the cell widths against the local scale $\sqrt{x^2+\nu_{\ell,d}}$, which allocates finer resolution near the origin and coarser resolution toward the endpoints. Equivalently, it amounts to uniform spacing in the transformed coordinate $u = \arsinh\!\left(\frac{x}{\sqrt{\nu_{\ell,d}}}\right)$. We therefore choose the Pauli-adapted hyperbolic alphabet
\begin{equation}\label{eq:hyperbolic-alphabet}
q_j=\sqrt{\nu_{\ell,d}}\,\sinh\!\left[\left(\frac{2j}{K-1}-1\right)A_{\ell,d}\right],\qquad j=0,\ldots,K-1.
\end{equation}
Note that the endpoints satisfy $q_0=-1$ and $q_{K-1}=1$. Applying \Cref{def:mean-preserving} to \eqref{eq:hyperbolic-alphabet} defines mean-preserving \textit{HyperQuant}, denoted by $Q_{K,\ell,d}$. The transmitted response is
\[
\widetilde Y_i=Q_{K,\ell,d}(\overline Y_i).
\]
The auxiliary random variables used by the quantizer are independent across batches and independent of the Pauli settings and quantum outcomes. By the tower property and the mean-preserving construction,
\[
\Ee[\widetilde Y_i\mid\bP_i]=\Ee[\overline Y_i\mid\bP_i]=\Tr(\bP_i\brhostar).
\]
Thus, the HyperQuant response remains conditionally unbiased for the underlying Pauli coefficient.  Before establishing the optimality properties, we examine the limiting case of a one-bit response.
\begin{remark}
When $K=2$, the alphabet is $\{\pm1\}$, and mean preservation uniquely determines $\mathbb P\bigl(Q(x)=1\mid x\bigr)=\frac{1+x}{2}$. Consequently,
$
\mathbb P\bigl(\widetilde Y_i=1\mid\bP_i\bigr)
=\frac{1+\Tr(\bP_i\brhostar)}{2}.
$
Thus, a one-bit HyperQuant response is conditionally equivalent to a single Pauli outcome and cannot retain the variance reduction from batch averaging, regardless of $\ell$.
\end{remark}
We now return to general $K$ and determine how the output levels
should be placed to control the additional quantization variance.
The following theorem shows that HyperQuant is
optimal, up to the stated factor, under a weighted minimax criterion.

\begin{theorem}\label{thm:weighted}
Let $K\ge2$, and define
\[
\mathfrak V_K(\nu_{\ell,d})=\inf_Q\sup_{x\in[-1,1]}\frac{\Var(Q(x)\mid x)}{x^2+\nu_{\ell,d}},
\]
where the infimum is over all mean-preserving $K$-level scalar quantizers with ordered alphabet $-1=q_0<\cdots<q_{K-1}=1$. Then
\begin{equation}\label{eq:weighted-minimax}\frac{A_{\ell,d}^2}{(K-1)^2}\le\mathfrak V_K(\nu_{\ell,d})\le\sup_{x\in[-1,1]}\frac{\Var(Q_{K,\ell,d}(x)\mid x)}{x^2+\nu_{\ell,d}}\le\exp\!\left(\frac{4A_{\ell,d}}{K-1}\right)\frac{A_{\ell,d}^2}{(K-1)^2}.
\end{equation}
\end{theorem}
The proof is deferred to \Cref{sec:proof-weighted}. In particular, \Cref{thm:weighted} shows that HyperQuant is minimax optimal up to an absolute constant whenever $K-1$ is a sufficiently large multiple of $A_{\ell,d}$.

The following elementary observation converts the weighted pointwise criterion in \Cref{thm:weighted} into a worst-case average distortion under a second-moment constraint. For an admissible quantizer $Q$, define
\[
R_Q(x):=\operatorname{Var}(Q(x)\mid x),\qquad T_Q:=\sup_{x\in[-1,1]}\frac{R_Q(x)}{x^2+\nu}.
\]
Then
\[\nu T_Q\le\sup_{\substack{\mu:\,\operatorname{supp}(\mu)\subset[-1,1]\\       \mathbb E_\mu X^2\le\nu}}\mathbb E_\mu R_Q(X)\le 2\nu T_Q.
\]
Indeed, the upper bound follows immediately from $R_Q(x)\le (x^2+\nu)T_Q$. For the lower bound, choose a point $x$ arbitrarily close to attaining $T_Q$. If $x^2\le\nu$, use $\mu=\delta_x$; otherwise use $\mu=\frac{\nu}{x^2}\delta_x+\left(1-\frac{\nu}{x^2}\right)\delta_0$.

\begin{corollary}\label{cor:moment-minimax}
Let the supremum below range over all probability distributions $\mu$ supported on $[-1,1]$ satisfying $\Ee_{X\sim\mu}X^2\le\nu_{\ell,d}$,  and let the infimum range over the same class of quantizers as in \Cref{thm:weighted}. Then
\[
\nu_{\ell,d}\mathfrak V_K(\nu_{\ell,d})\le\inf_Q\sup_{\mu}\Ee_{X\sim\mu}\Var(Q(X)\mid X)\le 2\nu_{\ell,d}\mathfrak V_K(\nu_{\ell,d}).
\]
\end{corollary}

For the specific Pauli response, take $X=\overline{Y}_i$ and $Q=Q_{K,\ell,d}$. Since HyperQuant is mean-preserving, $\mathbb E\!\left[(\widetilde Y_i-\overline{Y}_i)^2 \mid \overline{Y}_i\right]=\operatorname{Var}\!\left(Q_{K,\ell,d}(\overline{Y}_i)\mid \overline{Y}_i\right)$. Moreover, \eqref{eq:pauli-batch-second-moment} gives $\mathbb E \overline{Y}_i^2\le\nu_{\ell,d}$.  Applying the preceding comparison for the fixed quantizer $Q_{K,\ell,d}$, together with the pointwise upper bound in \Cref{thm:weighted}, yields 
\begin{equation}\label{eq:quant-var}
\mathbb E(\widetilde Y_i-\overline{Y}_i)^2\le 2\exp\!\left(\frac{4A_{\ell,d}}{K-1}\right)\frac{\nu_{\ell,d}A_{\ell,d}^2}{(K-1)^2}.
\end{equation}
The estimate in \eqref{eq:quant-var} is distribution-free: it uses the Pauli model only through the second-moment envelope \eqref{eq:pauli-batch-second-moment}.  It is natural to ask whether its dependence on $d$ can be improved by exploiting the finer distributional structure of Pauli coefficients. To answer this question, we consider the high-shot limit, in which the batch average approaches the exact $\Tr(\bP\brho)$. The next result shows that the dimension dependence in \eqref{eq:quant-var} is sharp, up to absolute constants, in the worst case over pure states.

\begin{theorem}\label{thm:pauli-distortion}
Suppose that $d$ is sufficiently large and $K\ge2\arsinh(\sqrt d)+1$. Then there exists an absolute constant $C_1>0$ such that
\begin{equation}\label{eq:pauli-distortion}
\frac{C_1\log^2d}{d(K+1)^2}\le\inf_Q\sup_{\brho\in\D_1}\Ee_{\bP}\Var\!\left(Q(\Tr(\bP\brho))\mid\Tr(\bP\brho)\right)\le\frac{2e^2\log^2d}{d(K-1)^2},
\end{equation}
where the infimum is over all $K$-level scalar quantizers with ordered alphabet $-1=q_0<\cdots<q_{K-1}=1$ that satisfy the conditional-mean constraint.
\end{theorem}

The proof is given in \Cref{sec:proof-pauli-distortion}. Therefore, in the high-shot pure-state regime, the minimax scalar distortion is of order $\frac{\log^2 d}{dK^2}$. Together with \Cref{thm:weighted}, this shows that HyperQuant captures the correct dependence on both the alphabet size $K$ and the Pauli dimension $d$. We next use these quantization guarantees to analyze state recovery from the finite-bit responses.

\section{Rank-Constrained Least-Squares Recovery}\label{sec:least-squares}
In this section, we study the rank-constrained least-squares estimator based on the quantized responses. We first show that the HyperQuant observations $\widetilde Y_i$ can be used directly in a quadratic loss without changing the population target. Consider the empirical least-squares loss
\begin{equation}\label{eq:loss}
\mathcal L_M(\bX)=\frac d{2M}\sum_{i=1}^M\bigl(\Tr(\bP_i\bX)-\widetilde Y_i\bigr)^2.
\end{equation}
For any $\bX\in\D_r$, we have
\begin{align*} 
&~\bigl(\Tr(\bP_i\bX)-\widetilde Y_i\bigr)^2-\bigl(\Tr(\bP_i\brhostar)-\widetilde Y_i\bigr)^2\cr
=&~\Tr(\bP_i(\bX-\brhostar))^2+2\Tr(\bP_i(\bX-\brhostar))\bigl(\Tr(\bP_i\brhostar)-\widetilde Y_i\bigr).
\end{align*}
Conditioning on $\bP_i$, the mean-preserving property of HyperQuant implies that the second term has zero expectation. Moreover, Pauli orthogonality gives $\Ee_{\bP}\Tr(\bP(\bX-\brhostar))^2=\frac1d\norm{\bX-\brhostar}_{\F}^2$. Averaging over the $M$ batches, we obtain
\begin{equation*}
\Ee\bigl[\mathcal L_M(\bX)-\mathcal L_M(\brhostar)\bigr]=\frac12\norm{\bX-\brhostar}_{\F}^2.
\end{equation*}
Thus, although quantization changes the fluctuations of the observations, it does not change the population least-squares objective: the expected excess loss is exactly one half the squared Frobenius distance from $\brhostar$. This motivates the global rank-constrained least-squares estimator
\begin{equation}\label{eq:ls-estimator}
\widehat{\brho}_{\rm LS}\in\argmin_{\bX\in\D_r}\mathcal L_M(\bX).
\end{equation}
We next quantify the finite-sample accuracy of this estimator. The following theorem shows that every global minimizer of \eqref{eq:ls-estimator} is statistically accurate once the number of Pauli settings is sufficient for low-rank recovery.  The proof is deferred to \Cref{sec:proof-ls}. 

\begin{theorem}\label{thm:ls}
 Suppose that $d\ge2$, $\brhostar\in\D_r$, and responses are generated by HyperQuant. Then, there exists some absolute constant $C_2>0$ such that, with probability at least $1-d^{-10}$, every global minimizer $\widehat{\brho}_{\rm LS}$ of \eqref{eq:ls-estimator} satisfies
\begin{equation}\label{eq:ls-bound}
\norm{\widehat{\brho}_{\rm LS}-\brhostar}_{\F}\le C_2d\sqrt r\max\left\{\sqrt{\frac{V_{\ell,K,d}\log d}{M}},\frac{\log d}{M}\right\},
\end{equation}
provided $M\ge C_2 rd\log^7d$, where
\begin{equation}\label{eq:effective-variance}
V_{\ell,K,d}=\frac1\ell+2\exp \left[\frac{4A_{\ell,d}}{K-1}\right]\frac{\nu_{\ell,d}A_{\ell,d}^2}{(K-1)^2}.
\end{equation}
\end{theorem}

The variance proxy $V_{\ell,K,d}$
combines the finite-shot and quantization contributions, represented by the first and second terms in \eqref{eq:effective-variance}, respectively. The following remarks summarize the resulting bit--shot tradeoff and resource requirements.

\begin{remark}
\textit{(i) Bit--shot matching.}
To retain the statistical order of the unquantized estimator, it is sufficient to make the quantization contribution in $V_{\ell,K,d}$ no larger than the intrinsic shot-noise scale $1/\ell$. Since \(\nu_{\ell,d}=\frac1\ell\left(1+\frac{\ell-1}{d}\right), \) this is ensured, for a sufficiently large absolute constant $C_3>0$, by
\begin{equation}\label{eq:bit-shot-matching}
K\ge C_3A_{\ell,d}\sqrt{1+\frac{\ell-1}{d}}.
\end{equation}
Under \eqref{eq:bit-shot-matching}, we have $V_{\ell,K,d}\le 2/\ell$. Hence, under the sample-size in  \Cref{thm:ls}, with probability at least $1-d^{-10}$, $\norm{\widehat{\brho}_{\rm LS}-\brhostar}_{\F}\le \sqrt2\,C_2d\sqrt r \max\left\{\sqrt{\frac{\log d}{\ell M}},\frac{\log d}{M}\right\}$. Moreover, \eqref{eq:bit-shot-matching} is satisfied by
\[
K\gtrsim\begin{cases}\log(1+\sqrt{\ell}), & \ell\le d,\\[0.3em]
\sqrt{\ell/d}\,\log(ed), & \ell>d, \end{cases}
\quad\text{~or~}\quad  b=
\begin{cases}
\cO(\log\log(e+\ell)), & \ell\le d,\\[0.3em] \frac12\log_2(\ell/d)+\cO(\log\log(ed)), & \ell>d.
\end{cases}\]
Thus, when $\ell\le d$, only logarithmically many quantization levels are needed, whereas for $\ell>d$ the required alphabet size grows as $\sqrt{\ell/d}$ up to a logarithmic factor.

\textit{(ii) Accuracy-constrained resource allocation.} For a target Frobenius accuracy $\varepsilon$, under \eqref{eq:bit-shot-matching}, it is sufficient to take $M\gtrsim\max\left\{rd\log^7d,\frac{d\sqrt r\log d}{\varepsilon}\right\},  M\ell\gtrsim\frac{rd^2\log d}{\varepsilon^2}$. Thus, in the accuracy-limited regime $M\lesssim rd^2\log d/\varepsilon^2$, choosing  $\ell \asymp \frac{rd^2\log d}{M\varepsilon^2}$ yields $N=M\ell=\widetilde{\cO}\left(\frac{rd^2}{\varepsilon^2}\right)$. For fixed admissible $M$, the corresponding bit budget is
\[B=Mb=\begin{cases}
\cO\!\left(M\log\log(e+\ell)\right), & \ell\le d,\\[0.4em]
\cO\!\left(M\left[\frac12\log(\ell/d)+\log\log(ed)\right]\right), & \ell>d.
\end{cases}\]
For comparison, lossless fixed-length transmission of the unquantized batch counts requires $B_{\rm lossless}=M\left\lceil\log_2(\ell+1)\right\rceil$. In particular, when $\ell=\Theta(d)$, the response budget is reduced from $\cO(M\log d)$ to $\cO(M\log\log d)$.
\end{remark}

\begin{remark}
\textit{Comparison with uniform quantization.} For a uniform $K$-level stochastic quantizer on $[-1,1]$, the quantization variance is bounded by $(K-1)^{-2}$. Matching this with the shot-noise scale $1/\ell$ therefore requires $K\gtrsim\sqrt{\ell}$. In contrast, HyperQuant requires only $K\gtrsim\log(1+\sqrt{\ell})$ when $\ell\le d$, and $K\gtrsim\sqrt{\ell/d}\log(ed)$ when $\ell>d$, while retaining the same order of least-squares accuracy.
\end{remark}

\section{Riemannian Optimization for Quantized QST}\label{sec:rgd}

The least-squares results in \Cref{sec:least-squares} characterize a minimizer but do not provide an efficient procedure for computing one. We therefore develop a Riemannian gradient method for quantized QST that directly optimizes the same least-squares objective formed from the quantized responses, with provable convergence guarantees. The proposed algorithm is dubbed \textit{QuantRGD} and summarized in \Cref{alg:rgd}.

\begin{algorithm}[H]
\caption{\textbf{Quant}ized QST via \textbf{R}iemannian \textbf{G}radient \textbf{D}escent (QuantRGD)
}\label{alg:rgd}
\begin{algorithmic}[1]
\REQUIRE $\{(\bP_i,\widetilde Y_i)\}_{i=1}^M$: quantized observations; $r$: rank.
\STATE $\brho_0=\mathcal H_r(\widetilde{\brho}_0)$, where $\widetilde{\brho}_0=\frac dM\sum_{i=1}^M\widetilde Y_i\bP_i$
\FOR{$k=0,1,\ldots$}
\STATE $\bG_k=\mathcal P_{T_{\brho_k}}\bigl(\nabla\mathcal L_M(\brho_k)\bigr)$
\STATE $\alpha_k=\norm{\bG_k}_{\F}^2/\norm{\A_M(\bG_k)}_2^2$
\STATE $\brho_{k+1}=\Pi_{\D_r}(\brho_k-\alpha_k\bG_k)$
\ENDFOR
\ENSURE $\widehat{\brho}=\brho_k$: recovered state.
\end{algorithmic}
\end{algorithm}

\subsection{Fixed-rank formulation and QuantRGD} The estimator in \eqref{eq:ls-estimator} involves a nonconvex rank constraint. For computation, we retain the same least-squares objective formed from the quantized responses and exploit its fixed-rank geometry,
\begin{equation*}
\min_{\bX\in\D_r}\ \mathcal L_M(\bX)=\frac{1}{2}\norm{\A_M(\bX)-\bm{y} }_2^2, \end{equation*}
where $\bm{y}=\sqrt{\frac dM}[\widetilde Y_1,\widetilde Y_2,\cdots,\widetilde Y_M ]^\top$ and the sampling operator is defined by
\begin{equation}\label{def:calA}
\A_M(\bX)=\sqrt{\frac dM}\bigl(\Tr(\bP_1\bX),\ldots,\Tr(\bP_M\bX)\bigr)^\top.
\end{equation}

We first introduce the fixed-rank Riemannian geometry.  Let $\mathbb H_d=\{\bZ\in\mathbb C^{d\times d}:\bZ=\bZ^\dagger\}$ denote the real Hilbert space of Hermitian matrices equipped with the Frobenius inner product, and define
\[
\mathcal M_r^{\rm H}=\{\bX\in\mathbb H_d:\rank(\bX)=r\},\qquad \mathcal M_r^+=\{\bX\in\mathbb H_d:\bX\succeq\bm 0,\ \rank(\bX)=r\}.
\]
For any $\bX\in\mathcal M_r^{\rm H}$ with compact eigendecomposition $\bX=\bU\bLambda\bU^\dagger$ and range projector $\bP_{\bX}=\bU\bU^\dagger$, the tangent space is $T_{\bX}\mathcal M_r^{\rm H}=\{\bU\bZ^\dagger+\bZ\bU^\dagger:\bZ\in\mathbb C^{d\times r}\}$.
The orthogonal projection of a Hermitian matrix $\bZ$ onto this tangent space has the closed form
\begin{equation}\label{eq:tangent-projection-hermitian}
\mathcal P_{T_{\bX}}(\bZ) =\bP_{\bX}\bZ+\bZ\bP_{\bX}-\bP_{\bX}\bZ\bP_{\bX}.
\end{equation}
With these geometric ingredients in place, QuantRGD proceeds as follows. Starting from $\brho_0 \in M_r^{\rm H}$, each iteration projects the Euclidean gradient onto the tangent space to obtain the Riemannian gradient, takes a descent step, and projects the resulting trial point onto $\D_r$ to enforce the density-matrix constraints. Specifically, given the current rank-$r$ state estimate $\brho_k$, the update is 
\begin{equation*} 
\bG_k=\mathcal P_{T_{\brho_k}} \bigl(\nabla\mathcal L_M(\brho_k)\bigr) \quad\textnormal{and}\quad \brho_{k+1}=\Pi_{\D_r}(\brho_k-\alpha_k\bG_k),
\end{equation*}
where $\bG_k$ is the Riemannian gradient and $\alpha_k>0$ is the stepsize. Since $\mathcal L_M$ is quadratic, exact line search along $-\bG_k$ gives
\begin{equation*} 
\alpha_k=\frac{\norm{\bG_k}_{\F}^2}{\norm{\A_M(\bG_k)}_2^2}.
\end{equation*}
The metric projection $\Pi_{\D_r}$ maps a Hermitian matrix to a nearest point in $\D_r$ under the Frobenius norm.  Let $\bW=\sum_{j=1}^d\lambda_j(\bW)\bu_j\bu_j^\dagger$ with $\lambda_1(\bW)\ge\cdots\ge\lambda_d(\bW)$. Then
\begin{equation}\label{eq:density-projection-spectral}
\Pi_{\D_r}\bW=\sum_{j=1}^r\mu_j\bu_j\bu_j^\dagger,\textnormal{ where } \bmu=\argmin_{\substack{\bz\in\mathbb R_+^r,\bone^\top\bz=1}}\norm{\bz-(\lambda_1(\bW),\ldots,\lambda_r(\bW))^\top}_2.
\end{equation}
Thus, $\Pi_{\D_r}$ can be computed by a leading $r$ Hermitian eigensolve followed by a simplex projection of the $r$ largest algebraic eigenvalues. 
The projection satisfies the following properties recorded in \Cref{lem:density-projection}, whose proof is given in \Cref{sec:proof_lmm41}.
\begin{lemma}\label{lem:density-projection} For every $\bW\in\mathbb H_d$ and $\brho\in\D_r$, we have
\begin{equation}\label{eq:density-projection-stability}
\norm{\Pi_{\D_r}(\bW)-\brho}_{\F}\le 2\norm{\bW-\brho}_{\F}.
\end{equation}
If, in addition, $\brho$ has rank $r$ and $\norm{\bW-\brho}_{\op}\le\frac{\lambda_r(\brho)}{4}$, then $\Pi_{\D_r}(\bW)$ also has rank $r$.
\end{lemma}

To initialize the nonconvex iteration, we form the spectral backprojection directly from the quantized responses, $\widetilde{\brho}_0= \A^*_M(\bm{y})=\frac dM\sum_{i=1}^M\widetilde Y_i\bP_i$, and set $\brho_0=\mathcal H_r(\widetilde{\brho}_0)$, where $\mathcal H_r$ denotes the rank-$r$ spectral truncation obtained by retaining the $r$ eigencomponents corresponding to the largest eigenvalues in magnitude.

\paragraph{Computational complexity} 
Each QuantRGD iteration costs $\cO(Mdr+dr^2+r^3)$ flops and requires $\cO(dr)$ working memory. Under the convergence conditions established in the following section, each $\bW_k:=\brho_k-\alpha_k\bG_k$ has at least $r$ positive eigenvalues, so its top $r$ eigenspace is contained in $\operatorname{range}(\bW_k)$.  Writing $\brho_k=\bU_k\bLambda_k\bU_k^\dagger$ and $\bZ_k=\nabla\mathcal L_M(\brho_k)$, we have $\operatorname{range}(\bW_k)\subseteq\operatorname{span}\{\bU_k,\bZ_k\bU_k\}$,  whose dimension is at most $2r$. Hence, after forming $\bZ_k\bU_k$, a thin QR factorization reduces the projection $\Pi_{\D_r}(\bW_k)$ to an eigendecomposition of a Hermitian matrix of order at most $2r$, followed by projection of its $r$ largest algebraic eigenvalues onto the probability simplex. Computing the Riemannian gradient and exact line search cost $\cO(Mdr)$ flops, while the projection onto $\D_r$ costs $\cO(dr^2+r^3)$ flops.  
\subsection{Convergence guarantees}
We establish convergence guarantees for QuantRGD under finite-bit Pauli responses. We first control the spectral initializer by bounding the operator-norm error of the backprojection and then converting it into a Frobenius-norm bound for its rank-$r$ truncation. 
\begin{theorem}\label{thm:init}
Suppose that $d\ge 2$, $\brhostar\in\D_r$ has rank $r$, and responses are generated by HyperQuant. There exists an absolute constant $C_4>0$ such that, with probability at least $1-d^{-10}$, we have 
\begin{equation}\label{eq:init-op}
\norm{\widetilde{\brho}_0-\brhostar}_{\op}
\le C_4\left[\sqrt{\frac{d\Tr(\brhostar^2)\log d}{M}}+d\sqrt{\frac{V_{\ell,K,d}\log d}{M}}+\frac{d\log d}{M}\right].
\end{equation}
If the right-hand side of \eqref{eq:init-op} is at most $\lambda_r(\brhostar)/4$, then $\brho_0=\mathcal H_r(\widetilde{\brho}_0)$ is positive semidefinite with rank $r$, and
\begin{equation}\label{eq:init-frobenius-bound}
\norm{\brho_0-\brhostar}_{\F}\le 2\sqrt{2r} \norm{\widetilde{\brho}_0-\brhostar}_{\op}.
\end{equation}
\end{theorem} 

Combining \Cref{thm:init} with deterministic contraction (proof in \Cref{sec:proof-local}) yields the following global recovery guarantee.

\begin{theorem}\label{thm:global recovery}
Suppose that $d\ge 2$, $\brhostar\in\D_r$ has rank $r$, responses are generated by HyperQuant, and the bit--shot condition \eqref{eq:bit-shot-matching} holds. Let $N=M\ell$. There exists an absolute constant $C_5>0$ such that if
\begin{align}
M&\ge C_5\max\left\{rd\log^7d,\frac{rd\Tr(\brhostar^2)\log d}{\lambda_r^2(\brhostar)},\frac{d\sqrt r\log d}{\lambda_r(\brhostar)}\right\},\label{eq:rgd-M-condition}\\
N&\ge C_5\frac{rd^2\log d}{\lambda_r^2(\brhostar)},\label{eq:rgd-N-condition}
\end{align}
then, with probability at least $1-2d^{-10}$, the iterates generated by \Cref{alg:rgd} satisfy
\begin{equation}\label{eq:global recovery}
\norm{\brho_k-\brhostar}_{\F}\le 4^{-k}\norm{\brho_0-\brhostar}_{\F}+C_5d\sqrt r\left[\sqrt{\frac{\log d}{N}}+\frac{\log d}{M}\right].
\end{equation}
\end{theorem}
The proofs of \Cref{thm:init,thm:global recovery} are deferred to \Cref{proofofthm:init,proofofglobalrecovery}, respectively. Thus, QuantRGD converges linearly to a statistical neighborhood of $\brhostar$, and the number of iterations required to reach this neighborhood is logarithmic in the ratio between the initialization error and the statistical radius. The conditions on $M$ ensure restricted isometry and control the random-design and bounded-increment terms, while the condition on $N$ controls the finite-shot and quantization contributions after bit--shot matching.  Since  $\Tr(\brhostar^2)\le r\lambda_1^2(\brhostar)$ and $\lambda_1(\brhostar)\ge 1/r$, the state-dependent setting requirement is at most $\widetilde \cO(\kappa^2r^2d)$, where $\kappa=\lambda_1(\brhostar)/\lambda_r(\brhostar)$, consistent with the scaling of the standard unquantized spectral-initialization analysis \citep{hsu2024rgd}. 

While \Cref{thm:ls} provides a statistical guarantee for every global least-squares minimizer, \Cref{thm:global recovery} provides global linear convergence of a computable two-stage method: the spectral initializer enters the contraction basin under explicit resource conditions, after which QuantRGD converges linearly to the statistical neighborhood.

\section{Numerical experiments}\label{sec:experiments}
We empirically examine the scalar distortion of HyperQuant, low-rank state recovery under a finite response-bit budget, and the runtime convergence of QuantRGD.  The response-level comparisons include the unquantized batch average $\overline{Y}$, the proposed HyperQuant response $Q^{\rm hyp}$, the uniform mean-preserving response $Q^{\rm unif}$, and direct raw-prefix transmission $Y_{\rm raw}^{(b)}$. Throughout the figures, $\overline{Y}$ denotes the full-batch average $\overline{Y}_i=\ell^{-1}\sum_{j=1}^{\ell}Y_{ij}$, encoded losslessly through its batch count.  The responses $Q^{\rm hyp}$ and $Q^{\rm unif}$ denote, respectively, the $K=2^b$-level HyperQuant and the uniform mean-preserving stochastic quantization of $\overline{Y}_i$, while $Y_{{\rm raw},i}^{(b)}$ averages only the first $\min\{b,\ell\}$ physical outcomes.  For the algorithmic comparisons, QuantRGD denotes \Cref{alg:rgd}, RGD denotes the rank-$r$ Riemannian gradient method of \citep{hsu2024rgd}, and MiFGD denotes the accelerated factorized method of \citep{kim2023fast}; the latter two use the unquantized response $\overline{Y}$.

 All experiments were run on an Apple M2 MacBook Air with 24\,GB of memory using MATLAB R2022b. Within each paired trial, all methods use the same state, Pauli settings, and physical outcomes.

\subsection{Scalar quantization}

We first test the scalar predictions independently of tomography. Panel (a) evaluates the weighted criterion $T_Q$ from \Cref{thm:weighted} for HyperQuant $Q^{\rm hyp}$ and uniform mean-preserving quantizer $Q^{\rm unif}$.

To examine the high-shot scaling in \Cref{thm:pauli-distortion}, we use the structured pure-state prior from its proof.  Write $d=2^n$, draw $k\sim\operatorname{Unif}\{2,\ldots,n-1\}$, set $m=2^k$, draw $\psi\sim\operatorname{Haar}(\mathbb C^m)$, and define $\brho_{k,\psi}=|\psi\rangle\langle\psi|\otimes |0\rangle\langle0|^{\otimes(n-k)}$. For a uniformly random Pauli matrix $\bP$, let $X=\Tr(\bP\brho_{k,\psi})$, and define the  Bayes distortion by
\[\mathcal D^{\rm str}_{d,K}:=\Ee_{k,\psi,\bP}\Var\!\left(Q^{\rm hyp}(X)\mid X\right),\]
where, in this high-shot experiment, $Q^{\rm hyp}$ uses the scale $\nu=1/d$.

Panel (a) uses $d=256$, $\ell=20d=5120$, and $b=1,\ldots,10$, with
$K=2^b$.  Panel (b) plots the normalized distortion $\frac{d(K-1)^2}{\log^2 d}\,\mathcal D^{\rm str}_{d,K}$ for $d\in\{2^5,\ldots,2^{16}\}$ and $b\in\{4,6,8\}$.  Each point in panel (b) is the mean of $20$ independent repetitions, each using $2\times10^5$ samples from the structured prior.  The dashed $K^{-2}$ curve in panel (a) indicates the predicted decay rate.

\Cref{fig:scalar_quantization}(a) shows that the weighted criterion for $Q^{\rm hyp}$ closely follows the predicted $K^{-2}$ decay.  At $K=16$, its value is approximately $20$ times smaller than that of $Q^{\rm unif}$.  \Cref{fig:scalar_quantization}(b) shows that the normalized Bayes distortion remains between approximately $0.5$ and $0.75$ over the tested dimensions and bit depths.  Its constant-order behavior and the close agreement across bit depths are consistent with the $\log^2 d/[d(K-1)^2]$ scaling in \Cref{thm:pauli-distortion}.

\begin{figure}
\centering
\subfloat[Weighted criterion \textit{vs.} $K=2^b$.]{\includegraphics[width=0.42\linewidth]{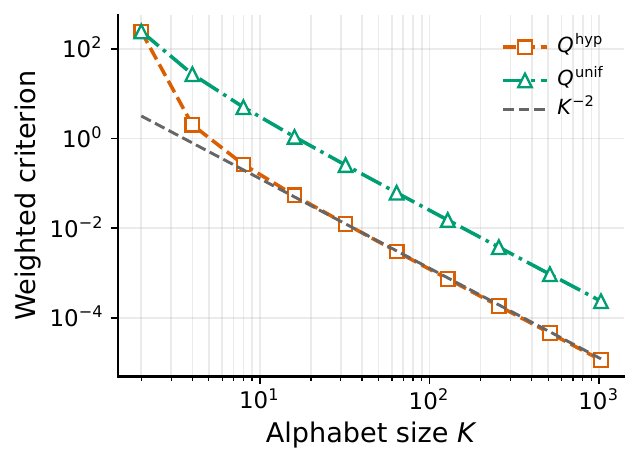}}
\qquad\qquad
\subfloat[Normalized pure-state distortion \textit{vs.} $d$.]{\includegraphics[width=0.42\linewidth]{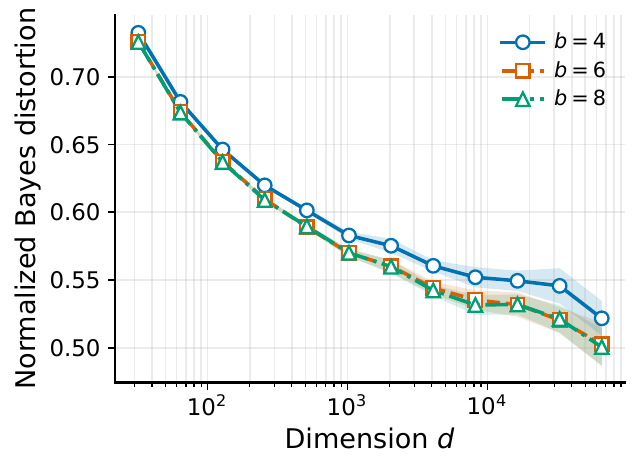}}
\caption{Scalar quantization.  Panel (a) plots the weighted variance criterion versus $K=2^b$ at $d=256$ and $\ell=20d$.  Panel (b) plots $d(K-1)^2\mathcal D^{\mathrm{str}}_{d,K}/\log^2d$ for $d=2^5,\ldots,2^{16}$ under the structured pure-state prior used in the proof of \Cref{thm:pauli-distortion}.  Curves are Monte Carlo means, and bands are 95\% confidence intervals over 20 independent repetitions.}
\label{fig:scalar_quantization}
\end{figure}

\subsection{Recovery performance}

We next examine how finite-bit response quantization affects low-rank state recovery.  In particular, we compare HyperQuant with the unquantized batch average, uniform mean-preserving quantization, and raw-prefix transmission under the same response-bit budget.

We use flat rank-three states with $d=2^7=128$ and $M=6rd=2304$. Panels (a) and (b) vary the number of copies using $b=4$ and $b=5$, respectively, while panel (c) fixes $\ell=20d=2560$ and varies $b=1,\ldots,7$.  All four response schemes are reconstructed using the same QuantRGD decoder, and each point summarizes $20$ paired trials. For trial $s$, define the quantization-to-batch-average error ratio
\[
R_s^{Q}(b,\ell) :=\frac{\|\widehat{\brho}^{(s)}_{Q}(b,\ell)-\brho_\star^{(s)}\|_{\F}}{\|\widehat{\brho}^{(s)}_{\overline Y}(\ell)-\brho_\star^{(s)}\|_{\F}}.
\]
with $Q\in\{Q^{{\rm hyp}},Q^{{\rm unif}},Y_{{\rm raw}}^{b}\}$.

\Cref{fig:recovery_performance} shows that HyperQuant remains close to the unquantized batch average and improves noticeably when the bit budget increases from four to five bits.  Panels (a) and (b) show that uniform quantization and raw-prefix transmission incur larger recovery errors at the same bit budget.  Panel (c) further shows that increasing $b$ closes the gap between HyperQuant and the unquantized response, whereas raw-prefix transmission remains less accurate because it uses only $b$ of the $\ell$ physical outcomes.

\begin{figure}[htp]
\centering
\subfloat[$b=4$: error \textit{vs.} $N$.]{\includegraphics[width=0.315\linewidth]{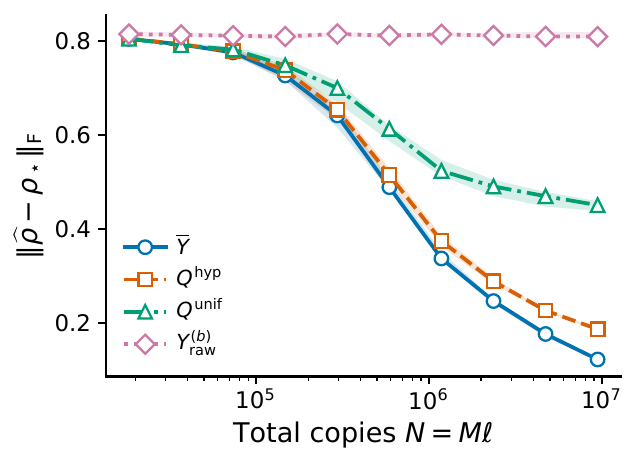}}
\hfill
\subfloat[$b=5$: error \textit{vs.} $N$.]{\includegraphics[width=0.315\linewidth]{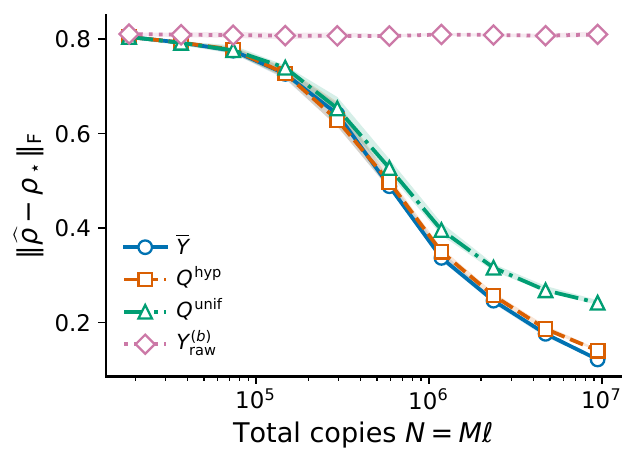}}
\hfill
\subfloat[Paired error ratio \textit{vs.} $b$.]{\includegraphics[width=0.315\linewidth]{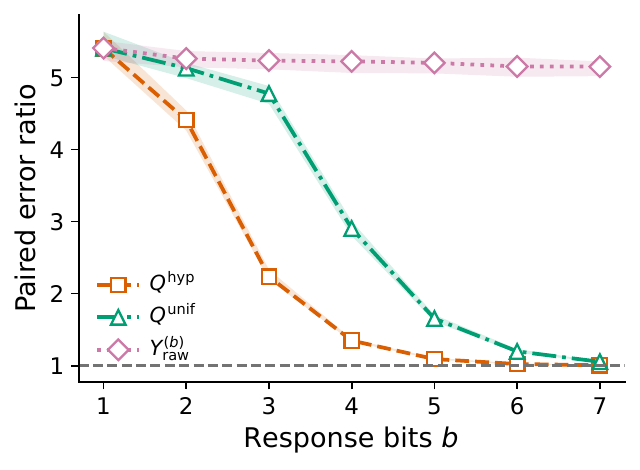}}
\caption{Recovery performance.  All panels use rank-three states with $d=2^7=128$ and $M=6rd=2304$.  Panels (a) and (b) plot reconstruction error versus $N=M\ell$ at $b=4$ and $b=5$, respectively; panel (c) plots the paired error ratio at $\ell=20d$.  Curves are medians with interquartile bands over 20 paired trials.}\label{fig:recovery_performance}
\end{figure}

\subsection{Convergence performance}
We examine the runtime convergence of QuantRGD with HyperQuant against RGD and MiFGD using the unquantized batch average $\overline Y$, together with their response-bit budgets.

All experiments use $r=3$, $M=6rd$, and $\ell=40d$.  Panels (a)--(d) consider $d\in\{2^5,2^7\}$ and $\kappa\in\{5,50\}$, with columns corresponding to $d=2^5,2^7$ and rows to $\kappa=5,50$.  The plotted error is $E_k:=\norm{\brho_k-\brhostar}_{\F}$. Within each paired trial, all methods share the target state, Pauli settings, and physical outcomes.  The algorithms run for at most $100$ iterations, and CPU time includes both initialization and iterative updates.  Since $\overline Y_i$ has $\ell+1$ possible batch counts, panel (e) uses response budgets $B_{\rm lossless} = M\left\lceil\log_2(\ell+1)\right\rceil$  for RGD and MiFGD, and $B=Mb$ for QuantRGD.

\begin{figure}[htp]
\centering
\subfloat[$d=2^5$ and $\kappa=5$.]{\includegraphics[width=0.33\linewidth]{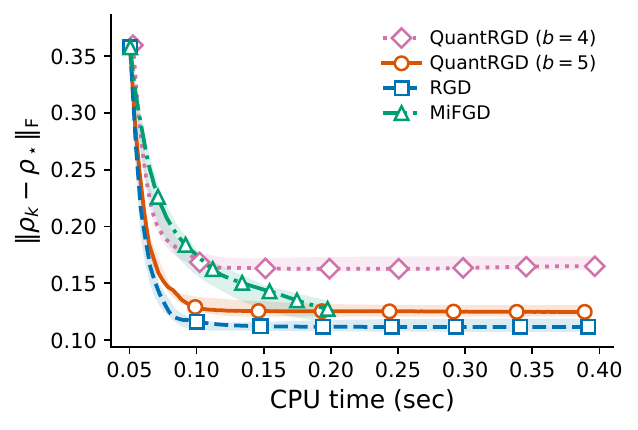}}
\subfloat[$d=2^7$ and $\kappa=5$.]{\includegraphics[width=0.33\linewidth]{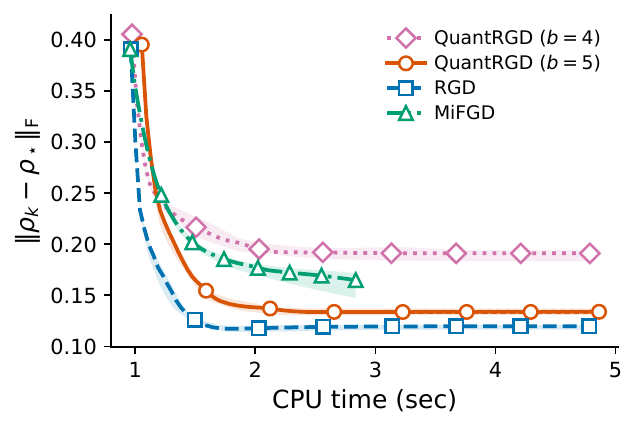}}
\makebox[0.33\linewidth]{}%
\par

\subfloat[$d=2^5$ and $\kappa=50$.]{\includegraphics[width=0.33\linewidth]{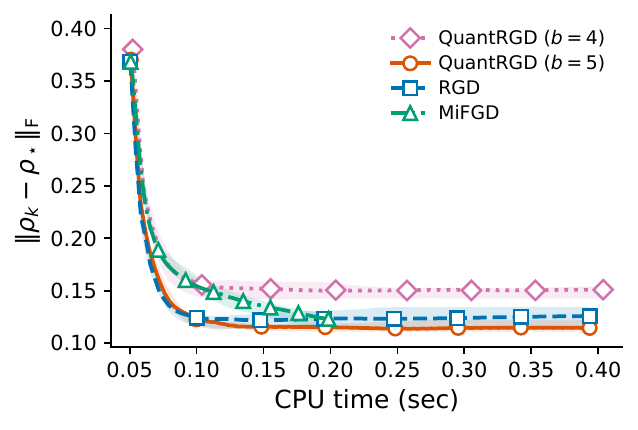}}
\subfloat[$d=2^7$ and $\kappa=50$.]{\includegraphics[width=0.33\linewidth]{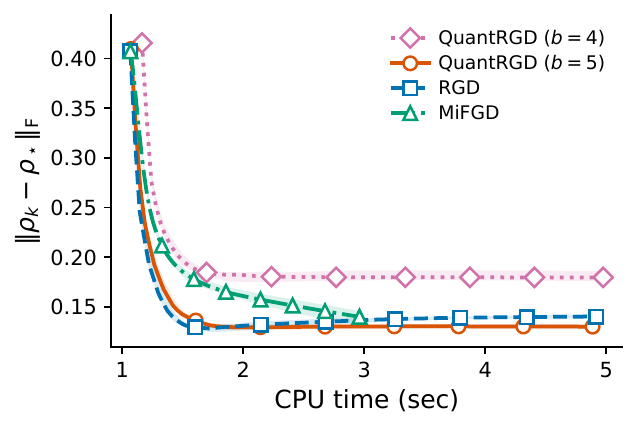}}
\subfloat[Total response budget $B$.]{\includegraphics[width=0.33\linewidth]{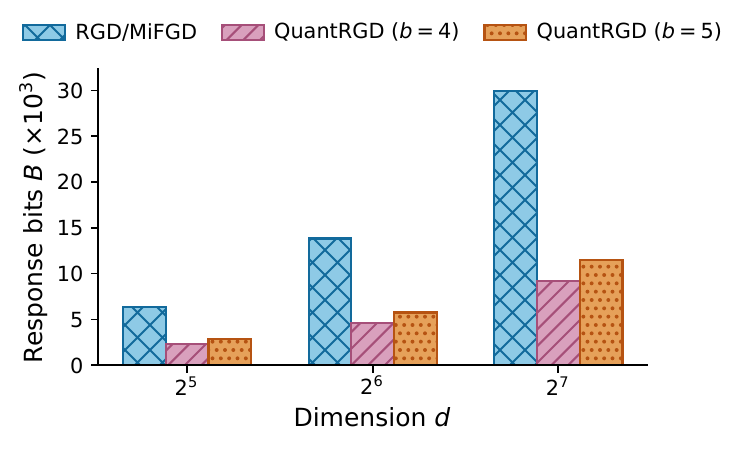}}

\caption{Convergence performance and response budget.  Panels (a)--(d) plot $\norm{\brho_k-\brhostar}_{\F}$ \textit{vs.} CPU time for $r=3$, $M=6rd$, and $\ell=40d$, with columns corresponding to $d=2^5,2^7$ and rows to $\kappa=5,50$.  Curves are medians with interquartile bands over ten paired trials.  Panel (e) compares the total response-bit budgets of unquantized RGD/MiFGD and QuantRGD with $b=4,5$.}
\label{fig:convergence_performance}
\end{figure}

\Cref{fig:convergence_performance} shows that QuantRGD with $b=5$ closely tracks unquantized RGD in both recovery error and runtime across the four tested settings, while $b=4$ reaches a higher finite-bit error floor. Panel (e) further shows that lossless transmission of $\overline Y$ requires approximately $2.2$--$2.6$ times the response-bit budget of QuantRGD with $b=5$ over the tested dimensions.

\section{Proofs of Main Results}\label{sec:proofs}\label{sec:proof-init} In this section, we provide the proofs of the quantization, least-squares, initialization, and convergence results. We follow the order in which the results appear in the main text. 

\subsection{Proof of \Cref{thm:weighted}}\label{sec:proof-weighted}
\begin{proof}[Proof of \Cref{thm:weighted}]
We prove the corresponding statement for an arbitrary scale $\nu \in (0,1]$ and then specialize to $\nu=\nu_{\ell,d}$.
Let
\[ A_\nu := \operatorname{arcsinh}(\nu^{-1/2}).\]
We first establish a pointwise lower bound that holds for every mean-preserving quantizer with a fixed ordered alphabet
\[-1=q_0<q_1<\cdots<q_{K-1}=1.\]
Fix $x\in[q_j,q_{j+1}]$ and let
\[    Z:=Q(x).\]
Since $Q$ is mean preserving, $\mathbb{E}[Z\mid x]=x$. Moreover, $Z$ takes values only in the alphabet $\{q_0,\ldots,q_{K-1}\}$. Since $q_j$ and $q_{j+1}$ are adjacent alphabet levels, every possible value of $Z$ satisfies either $Z\le q_j$ or $Z\ge q_{j+1}$. Hence $(Z-q_j)(Z-q_{j+1})\ge 0$. Taking conditional expectation and using $\mathbb{E}[Z\mid x]=x$ gives
\[
\begin{aligned}0 &\le \mathbb{E}\!\left[(Z-q_j)(Z-q_{j+1}) \mid x\right] =\mathbb{E}[Z^2\mid x]-(q_j+q_{j+1})x+q_jq_{j+1}.\end{aligned}
\]
Therefore,
\begin{equation}\label{eqn:opt_adj} \operatorname{Var}(Q(x)\mid x) = \mathbb{E}[Z^2\mid x]-x^2  \ge (q_j+q_{j+1})x-q_jq_{j+1}-x^2 = (x-q_j)(q_{j+1}-x). \end{equation}
Equality holds for adjacent stochastic rounding.

We next prove the upper bound. Consider the hyperbolic alphabet
\[q_j(\nu)=\sqrt{\nu}\,\sinh\!\left[\left(\frac{2j}{K-1}-1 \right)A_\nu\right],\qquad j=0,\ldots,K-1,\]
and let $Q_{K,\nu}$ denote adjacent stochastic rounding with respect to these levels. Write $x=\sqrt{\nu}\sinh y,  \tau:=\frac{2A_\nu}{K-1}$, and define $u_j:=-A_\nu+j\tau$. If $x\in[q_j(\nu),q_{j+1}(\nu)]$, then $y\in[u_j,u_j+\tau]$. Furthermore,
\[
\begin{aligned} q_{j+1}(\nu)-q_j(\nu) &=\sqrt{\nu}\int_{u_j}^{u_j+\tau}\cosh t\,dt. \end{aligned}
\]
For every $t\in[u_j,u_j+\tau]$, we have $|t-y|\le\tau$, and hence$\cosh t \le e^\tau \cosh y$. Consequently,
\[
\begin{aligned} q_{j+1}(\nu)-q_j(\nu) &\le\tau e^\tau \sqrt{\nu}\cosh y  =\tau e^\tau \sqrt{x^2+\nu}.\end{aligned}
\]
Since $Q_{K,\nu}$ uses adjacent stochastic rounding,
\[\operatorname{Var}(Q_{K,\nu}(x)\mid x)= (x-q_j)(q_{j+1}-x)\le\frac{1}{4}(q_{j+1}-q_j)^2.
\]
Thus, $\operatorname{Var}(Q_{K,\nu}(x)\mid x) \le \frac{\tau^2e^{2\tau}}{4}(x^2+\nu)$. Dividing by $x^2+\nu$ and substituting $\tau=2A_\nu/(K-1)$ yields
\[\sup_{x\in[-1,1]}\frac{\operatorname{Var}(Q_{K,\nu}(x)\mid x)}{x^2+\nu}\le\exp\!\left(\frac{4A_\nu}{K-1}\right)\frac{A_\nu^2}{(K-1)^2}.
\]
Since $\mathfrak V_K(\nu)$ is the infimum over all admissible quantizers, this also gives
\[
\mathfrak V_K(\nu)\le\sup_{x\in[-1,1]}\frac{\operatorname{Var}(Q_{K,\nu}(x)\mid x)}{ x^2+\nu}.
\]
It remains to prove the minimax lower bound. Fix an arbitrary admissible mean-preserving $K$-level quantizer $Q$ with alphabet $-1=q_0<q_1<\cdots<q_{K-1}=1$. Introduce the transformed alphabet locations
\[u_j:= \operatorname{arcsinh}\left( \frac{q_j}{\sqrt{\nu}}\right),\qquad  j=0,\ldots,K-1.
\]
Since $q_0=-1$ and $q_{K-1}=1$, we have $u_0=-A_\nu, u_{K-1}=A_\nu$.
Hence the transformed cell lengths satisfy
\(    \sum_{j=0}^{K-2}(u_{j+1}-u_j)=2A_\nu.
\) 
Therefore, there exists at least one index $j$ such that
\(
    u_{j+1}-u_j    \ge    \frac{2A_\nu}{K-1}.
\)
For this cell, write
\(   u_j=c-\delta,   u_{j+1}=c+\delta, \)
so that
\( \delta  = \frac{u_{j+1}-u_j}{2}  \ge  \frac{A_\nu}{K-1}.\) 
Equivalently,
\(
    q_j =  \sqrt{\nu}\sinh(c-\delta),   q_{j+1}  =  \sqrt{\nu}\sinh(c+\delta).\) 
Choose \(x:=\sqrt{\nu}\sinh c.\) Then $x\in[q_j,q_{j+1}]$. Applying~\eqref{eqn:opt_adj}, we obtain
\[\begin{aligned}    \frac{\operatorname{Var}(Q(x)\mid x)}        {x^2+\nu}   &\ge  \frac{(x-q_j)(q_{j+1}-x)}       {x^2+\nu} =   4\sinh^2\!\left(\frac{\delta}{2}\right)
    \frac{       \cosh(c-\delta/2)       \cosh(c+\delta/2)   }{        \cosh^2 c    }. \end{aligned} \]
Using the identity \(    \cosh(c-\delta/2)\cosh(c+\delta/2)=    \cosh^2 c+\sinh^2(\delta/2)    \ge \cosh^2 c,
\)
we obtain
\(    \frac{\operatorname{Var}(Q(x)\mid x)}       {x^2+\nu} \ge    4\sinh^2\!\left(\frac{\delta}{2}\right).
\)
Since $\sinh t\ge t$ for $t\ge0$,
\(  4\sinh^2\!\left(\frac{\delta}{2}\right)  \ge\delta^2\ge\frac{A_\nu^2}{(K-1)^2}.
\)
Thus, for every admissible quantizer $Q$,
\[
    \sup_{x\in[-1,1]}  \frac{\operatorname{Var}(Q(x)\mid x)}{x^2+\nu}   \ge   \frac{A_\nu^2}{(K-1)^2}.
\]
Taking the infimum over $Q$ gives
\(    \mathfrak V_K(\nu)    \ge    \frac{A_\nu^2}{(K-1)^2}.\)\\
Finally, setting $\nu=\nu_{\ell,d}$ gives $A_\nu=A_{\ell,d}$ and $Q_{K,\nu}=Q_{K,\ell,d}$, and therefore 
\[ \frac{A_{\ell,d}^2}{(K-1)^2}   \le    \mathfrak V_K(\nu_{\ell,d})    \le    \sup_{x\in[-1,1]}    \frac{        \operatorname{Var}(Q_{K,\ell,d}(x)\mid x)}{x^2+\nu_{\ell,d}}   \le \exp\!\left(\frac{4A_{\ell,d}}{K-1}\right) \frac{A_{\ell,d}^2}{(K-1)^2}.\]
\end{proof}
\subsection{Proof of \Cref{cor:moment-minimax}}
\begin{proof}[Proof of \Cref{cor:moment-minimax}]
For the lower bound, fix an arbitrary admissible quantizer $Q$ and let
\( T_Q  :=\sup_{x\in[-1,1]} \frac{\operatorname{Var}(Q(x)\mid x)}{x^2+\nu}. \)
We show that
\[   \sup_{\substack{\mu:\,\operatorname{supp}(\mu)\subset[-1,1]\\                  \mathbb E_\mu X^2\le\nu}}  \mathbb E_\mu\operatorname{Var}(Q(X)\mid X)\ge\nu T_Q.
\]

Let $\varepsilon>0$. By the definition of the supremum, there exists $x_\varepsilon\in[-1,1]$ such that
\[\frac{\operatorname{Var}(Q(x_\varepsilon)\mid x_\varepsilon)}{x_\varepsilon^2+\nu}\ge T_Q-\varepsilon.
\]
Equivalently,
$\operatorname{Var}(Q(x_\varepsilon)\mid x_\varepsilon)\ge (T_Q-\varepsilon)(x_\varepsilon^2+\nu)$.  We distinguish two cases. If $x_\varepsilon^2\le\nu$, take $\mu_\varepsilon=\delta_{x_\varepsilon}$. This distribution is admissible as $\mathbb E_{\mu_\varepsilon}X^2=x_\varepsilon^2\le\nu$. 
Hence, 
\[\begin{aligned} \mathbb E_{\mu_\varepsilon}\operatorname{Var}(Q(X)\mid X)&=\operatorname{Var}(Q(x_\varepsilon)\mid x_\varepsilon)\ge (T_Q-\varepsilon)(x_\varepsilon^2+\nu)\ge\nu(T_Q-\varepsilon).
\end{aligned} \]
If $x_\varepsilon^2>\nu$, the point mass $\delta_{x_\varepsilon}$ is not admissible because its second moment exceeds $\nu$. Instead, define
\[ \mu_\varepsilon:= \frac{\nu}{x_\varepsilon^2}\delta_{x_\varepsilon}+\left(1-\frac{\nu}{x_\varepsilon^2} \right)\delta_0. \]
Since $x_\varepsilon^2>\nu$, both coefficients are nonnegative and sum to one. Moreover,
\(\begin{aligned} \mathbb E_{\mu_\varepsilon}X^2&=\frac{\nu}{x_\varepsilon^2} x_\varepsilon^2 +\left(1-\frac{\nu}{x_\varepsilon^2}\right)0=\nu,
\end{aligned}\)
so $\mu_\varepsilon$ is admissible. Since conditional variances are nonnegative,
\[
\begin{aligned}  \mathbb E_{\mu_\varepsilon}  \operatorname{Var}(Q(X)\mid X) &= \frac{\nu}{x_\varepsilon^2} \operatorname{Var} (Q(x_\varepsilon)\mid x_\varepsilon) +\left(  1-\frac{\nu}{x_\varepsilon^2}\right) \operatorname{Var}(Q(0)\mid 0)\cr
    &\ge\frac{\nu}{x_\varepsilon^2} \operatorname{Var}(Q(x_\varepsilon)\mid x_\varepsilon)\ge\frac{\nu}{x_\varepsilon^2}(T_Q-\varepsilon)(x_\varepsilon^2+\nu)\cr
    &= \nu(T_Q-\varepsilon)\left(  1+\frac{\nu}{x_\varepsilon^2}\right)\ge \nu(T_Q-\varepsilon).
\end{aligned}\]
Therefore, in either case,  \[ \sup_{\substack{\mu:\,\operatorname{supp}(\mu)\subset[-1,1]\\ \mathbb E_\mu X^2\le\nu}} \mathbb E_\mu \operatorname{Var}(Q(X)\mid X) \ge \nu(T_Q-\varepsilon). \]
Since this holds for every $\varepsilon>0$, taking the limit as $\varepsilon$ tends to zero yields
\[ \sup_{\substack{\mu:\,\operatorname{supp}(\mu)\subset[-1,1]\\ \mathbb E_\mu X^2\le\nu}} \mathbb E_\mu \operatorname{Var}(Q(X)\mid X) \ge \nu T_Q. \] 
Finally, because it holds for every admissible $Q$, taking the infimum over $Q$ gives
\[
\begin{aligned} \inf_Q\sup_{\substack{\mu:\,\operatorname{supp}(\mu)\subset[-1,1]\\
  \mathbb E_\mu X^2\le\nu}} \mathbb E_\mu\operatorname{Var}(Q(X)\mid X)  &\ge\nu\inf_Q T_Q=\nu\mathfrak V_K(\nu).
\end{aligned}\]
Combining the above lower-bound argument with the upper-bound argument given immediately before \Cref{cor:moment-minimax}, setting $\nu=\nu_{\ell,d}$, and taking the infimum over all admissible quantizers $Q$ yields the claimed result.
\end{proof}

\subsection{Proof of \Cref{thm:pauli-distortion}}\label{sec:proof-pauli-distortion}
\begin{proof}[Proof of \Cref{thm:pauli-distortion}] For the upper bound, repeat the upper-bound construction in the proof of \Cref{thm:weighted} with $\nu=1/d$. Pauli orthogonality, $\Tr(\brho^2)\le1$, and the condition $2\arsinh(\sqrt d)/(K-1)\le1$ give
\[\Ee_{\bP}\Var(Q(\Tr(\bP\brho))\mid\Tr(\bP\brho))\le 2e^2\frac{\arsinh^2(\sqrt d)}{d(K-1)^2}. \]
Since $\arsinh(\sqrt d)\le\log d$ for $d\ge16$, the upper bound in \eqref{eq:pauli-distortion} follows.

We prove the lower bound by a minimax argument. Fix an arbitrary
mean-preserving $K$-level scalar quantizer $Q$. We will construct a
pure state $\rho\in\mathcal D_1$ such that
\(
    \mathbb E_P \operatorname{Var} \left(Q(\operatorname{Tr}(P\rho)) \mid \operatorname{Tr}(P\rho) \right)\gtrsim \frac{\log^2 d}{d(K+1)^2}.\)
Since the desired quantity is a supremum over pure states, it is enough to construct a prior over pure states for which the average distortion is large. Indeed, for any probability distribution $\Pi$ supported on $\mathcal D_1$,
\[
    \sup_{\rho\in\mathcal D_1} \mathbb E_P\operatorname{Var} \left( Q(\operatorname{Tr}(P\rho)) \mid \operatorname{Tr}(P\rho)  \right)  \ge\mathbb E_{\rho\sim\Pi}\mathbb E_P\operatorname{Var}\left(  Q(\operatorname{Tr}(P\rho))  \mid  \operatorname{Tr}(P\rho)\right).
\]
The main idea is to choose a multiscale prior. A Haar-random pure state in dimension $m$ produces Pauli coefficients of typical size $m^{-1/2}$. By mixing over dyadic dimensions \(  m=2^k,k=2,\ldots,n-1\),  we create a distribution of Pauli coefficients spanning all scales between $d^{-1/2}$ and a constant. This prevents any $K$-level quantizer from allocating all of its resolution to a single scale.

Let \( k\sim\operatorname{Unif}\{2,\ldots,n-1\},   m=2^k, \) and, conditioned on $k$, draw a Haar-distributed unit vector $\psi\in\mathbb C^m$. Define the $n$-qubit pure state \( \rho_{k,\psi}    :=|\psi\rangle\langle\psi| \otimes  |0\rangle\langle0|^{\otimes(n-k)}. \)
Let \(X=\operatorname{Tr}(P\rho_{k,\psi}),\) where $P$ is uniformly sampled from all $n$-qubit Pauli matrices. 

 We first characterize the distribution of $X$. Fix a nonidentity Pauli matrix $P_A$ acting on $\mathbb C^m$. Since $P_A$ has eigenvalues $\pm1$ of equal  multiplicity $m/2$, Haar invariance allows us to assume, without loss of generality, that
\[
    P_A= \begin{pmatrix} I_{m/2} & 0\\  0 & -I_{m/2}\end{pmatrix}.
\]
Let $\Pi_+$ and $\Pi_-$ denote the projectors onto the corresponding
eigenspaces, and define
\(   S:=\|\Pi_+\psi\|_2^2  =  \sum_{j=1}^{m/2}|\psi_j|^2. \) For a Haar-distributed unit vector $\psi\in\mathbb C^m$, the squared norm of the projection onto an $m/2$-dimensional subspace follows the standard Beta distribution:
\(  S\sim  \operatorname{Beta}\left(\frac m2,\frac m2\right) \) see \textit{e.g.}, \citep{mezzadri2007,zyczkowski2001induced}.

Since the eigenvalues of $P_A$ are $\pm1$, \(\langle\psi,P_A\psi\rangle  = \|\Pi_+\psi\|_2^2-\|\Pi_-\psi\|_2^2= S-(1-S)=2S-1.\)  Therefore,   the density of \( X=\langle\psi,P_A\psi\rangle \) is
\begin{equation}
\label{eqn:distributionofgmx}
g_m(x) =\frac{\Gamma(m)}{2^{m-1}\Gamma(m/2)^2}(1-x^2)^{m/2-1},  \qquad -1<x<1.
\end{equation}

 We next account for the ancilla qubits. Write
\(  P=P_A\otimes P_B .\) Then \(  X   =\langle\psi,P_A\psi\rangle   \left\langle  0^{\otimes(n-k)},    P_B0^{\otimes(n-k)}   \right\rangle .\) The ancilla factor is nonzero only when every tensor factor of $P_B$ is either $I$ or $\sigma_z$. Hence the fraction of Pauli strings with a nonzero ancilla contribution is
\(  \frac{2^{n-k}}{4^{n-k}} =\frac md .\)
Since $k$ is uniformly distributed over $\{2,\ldots,n-1\}$, the continuous component of the mixture density contains
\begin{equation}   \frac1{n-2}\frac md \left(1-\frac{1}{m^2}\right)g_m(x).    \label{eq:mixture-component} \end{equation}

We lower bound this density. For \(    x\in\left[\frac2{\sqrt d},\frac12\right], \)
choose a dyadic $m\in\{4,8,\ldots,\frac{d}{2}\}$ such that
\(  1\le mx^2<2. \) By Stirling's inequalities and the definition of $g_m$, \( g_m(x)\ge c\sqrt m .\)
Therefore, using \eqref{eq:mixture-component}, we have
\( f(x)\ge \frac{c}{n-2}\frac md\sqrt m = \frac{c m^{3/2}}{d(n-2)} .\)
Moreover, $1\le mx^2<2$ implies \( m^{3/2}\ge x^{-3}.\)
Hence,
\begin{equation}
    f(x) \ge \frac{c}{d(n-2)x^3},\qquad \frac2{\sqrt d}\le x\le\frac12 .
    \label{eq:mixture-density-lower}
\end{equation}
We now lower bound the distortion of the arbitrary quantizer $Q$. Let its alphabet be
\[  -1=q_0<q_1<\cdots<q_{K-1}=1.\]
Intersect the quantization cells with the interval \(  \left[\frac2{\sqrt d},\frac12\right],\) and denote the nonempty intersections by \( [a_s,b_s], s=1,\ldots,J.\)
Clearly, \(J\le K+1.\)  For $x\in[a_s,b_s]$, the pointwise lower bound from \eqref{eqn:opt_adj} gives
\begin{equation}\label{eqn:lower_VarQ}
    \operatorname{Var}(Q(x)\mid x)\ge(x-a_s)(b_s-x).
 \end{equation}
Combining \eqref{eq:mixture-density-lower} and \eqref{eqn:lower_VarQ} yields \( \mathbb E  \operatorname{Var}(Q(X)\mid X)  \ge  \frac{c_5}{d(n-2)} \sum_{s=1}^J \int_{a_s}^{b_s} \frac{(x-a_s)(b_s-x)} {x^3}\,dx.\)  For any $0<a<b$, direct integration gives
\[\begin{aligned} \int_a^b \frac{(x-a)(b-x)}{x^3}\,dx  &=  \sinh\left(\log\frac ba\right) - \log\frac ba .
\end{aligned}\]
Since
\(\sinh(t)-t\ge\frac{t^3}{6}, t\ge0,\) we have \(  \int_a^b  \frac{(x-a)(b-x)}{x^3}\,dx \ge \frac16  \left(\log\frac ba\right)^3 .\) Define the logarithmic widths \(  L_s:=\log\frac{b_s}{a_s}.\)
Because the intervals $[a_s,b_s]$ partition
\( \left[\frac2{\sqrt d},\frac12\right],\) 
\(\sum_{s=1}^J L_s  =\log\frac{\sqrt d}{4}. \) By Jensen's inequality,
\[\sum_{s=1}^J L_s^3\ge\frac{\left(\sum_{s=1}^J L_s\right)^3}{J^2}\ge\frac1{(K+1)^2}\left(\log\frac{\sqrt d}{4}\right)^3.\]
Therefore, 
\(\mathbb E\operatorname{Var}(Q(X)\mid X)\ge\frac{c_6}{d(n-2)(K+1)^2}\left( \log\frac{\sqrt d}{4}\right)^3.\) Since $d=2^n$,
\( n-2\asymp\log d,\) and \( \log(\sqrt d/4)\asymp\log d\) for sufficiently large $d$. Hence, \[\mathbb E\operatorname{Var}(Q(X)\mid X)\ge C_1 \frac{\log^2 d}{d(K+1)^2}.\]
This lower bound holds for every admissible quantizer $Q$. Taking the infimum over $Q$ completes the proof.
\end{proof}

\subsection{Proof of \Cref{thm:ls}}\label{sec:proof-ls}
For $s\ge1$, let $\delta_s$ denote the rank-$s$ restricted isometry constant of $\A_M$. The proof relies on \Cref{lem:pauli-rip}, which establishes the Pauli RIP, and \Cref{lem:score}, which controls the score at the true state. 

\begin{lemma}\cite[Theorem~2.1]{liu2011universal}\label{lem:pauli-rip}
There exists an absolute constant $c_4>0$ such that the following holds. Suppose that $1\le s\le d$, $\delta\in(0,1)$, and $\eta\in(0,1)$. If
\begin{equation}\label{eq:pauli-rip-sample-proof} 
M\ge c_4\delta^{-2}sd\log^6(2d)\log(2/\eta),
\end{equation}
then, with probability at least $1-\eta$,
\begin{equation*}
(1-\delta)\norm{\bH}_{\F}^2\le\norm{\A_M(\bH)}_2^2\le(1+\delta)\norm{\bH}_{\F}^2
\end{equation*}
holds simultaneously for all Hermitian matrices $\bH$ of rank at most $s$.
\end{lemma}

\begin{lemma}\label{lem:score}
There exists an absolute constant $c_5>0$ such that the following holds. Suppose that $d\ge2$ and $\eta\in(0,1)$. Then, with probability at least $1-2\eta$,
\begin{align}
\norm{\nabla\mathcal L_M(\brhostar)}_{\op}&\le c_5d\max\left\{\sqrt{\frac{V_{\ell,K,d}\log(d/\eta)}M},\frac{\log(d/\eta)}M\right\},\label{eq:score-op}\\
\norm{\nabla\mathcal L_M(\brhostar)}_{(2r)}&\le c_5d\sqrt r\max\left\{\sqrt{\frac{V_{\ell,K,d}\log(d/\eta)}M},\frac{\log(d/\eta)}M\right\}.\label{eq:score-restricted}
\end{align}
\end{lemma}
\begin{proof}
We first note that, for every Hermitian matrix $\bZ$ and $s\ge1$,
\begin{equation}
    \norm{\bZ}_{(s)} =\Big(\sum_{j=1}\sigma_j^2(\bZ)\Big)^{1/2}  \le \sqrt{s}\,\norm{\bZ}_{\op}. \label{eq:restricted-op}
\end{equation}

Set
\( \xi_i:=  \widetilde Y_i-\Tr(\bP_i\brhostar).\) By the mean-preserving property, \(  \Ee[\xi_i\mid\bP_i]=0.\) Moreover, \(\xi_i =\bigl(Y_i-\Tr(\bP_i\brhostar)\bigr)   +\bigl(\widetilde Y_i-Y_i\bigr).\) Since \(  \Ee[\widetilde Y_i-Y_i\mid Y_i,\bP_i]=0, \)
the cross term vanishes, and therefore
\[ \Ee\xi_i^2 =\Ee\bigl(Y_i-\Tr(\bP_i\brhostar)\bigr)^2 +\Ee(\widetilde Y_i-Y_i)^2 \le\frac1\ell +2\exp\left( \frac{4A_{\ell,d}}{K-1}\right) \frac{\nu_{\ell,d}A_{\ell,d}^2}{(K-1)^2}= V_{\ell,K,d}.\] 
Also, since both $\widetilde Y_i$ and $\Tr(\bP_i\brhostar)$ belong to $[-1,1]$, \(|\xi_i|\le2.\) Define \(\bX_i:=-\frac dM\xi_i\bP_i.\) The matrices $\bX_i$ are independent, mean-zero, and Hermitian. Moreover, \( \norm{\bX_i}_{\op}\le\frac{2d}{M}.\) Since $\bP_i^2=\bI$, \(\bX_i^2=\frac{d^2}{M^2}\xi_i^2\bI,\) and hence
\[\sum_{i=1}^M\Ee\bX_i^2\preceq M^{-1}d^2V_{\ell,K,d}\bI.\]
Thus \(\left\|\sum_{i=1}^M\Ee\bX_i^2   \right\|_{\op}\le \frac{d^2V_{\ell,K,d}}{M}.\)  Applying the Hermitian matrix Bernstein inequality gives, with probability at least $1-2\eta$,
\[    \left\|      \sum_{i=1}^M\bX_i \right\|_{\op} \le c_5d \max\left\{   \sqrt{\frac{V_{\ell,K,d}\log(d/\eta)}{M}}, \frac{\log(d/\eta)}{M}  \right\}.\]
On the other hand,  \( \sum_{i=1}^M\bX_i =\frac dM\sum_{i=1}^M \bigl(\Tr(\bP_i\brhostar)-\widetilde Y_i \bigr)\bP_i = \nabla\mathcal L_M(\brhostar),\) which proves \eqref{eq:score-op}.
Finally, applying \eqref{eq:restricted-op} with $s=2r$ gives 
\(\norm{\nabla\mathcal L_M(\brhostar)}_{(2r)}  \le\sqrt{2r}\, \norm{\nabla\mathcal L_M(\brhostar)}_{\op}.\) After enlarging the absolute constant $c_5$ if necessary, this yields \eqref{eq:score-restricted}.
\end{proof}

\begin{proof}[Proof of \Cref{thm:ls}]
The proof combines the restricted isometry of the Pauli sampling operator with the concentration bound for the score at the true state.  We first apply \Cref{lem:pauli-rip} with
\(s=\min\{2r,d\}, \delta=\frac12, \eta=d^{-11}.\)
Since $s\le 2r$, $\log(2d)\le 2\log d$, and $\log(2d^{11})\le 12\log d$ for $d\ge2$,   \eqref{eq:pauli-rip-sample-proof} follows from
\(M\ge C_2rd\log^7 d\)
whenever $C_2$ is sufficiently large. Consequently, with probability at least $1-d^{-11}$, \begin{equation}
\|\A_M(\bH)\|_2^2\ge\frac12\norm{\bH}_{\F}^2 \label{eq:ls-rip-event}
\end{equation}
holds simultaneously for every Hermitian matrix $\bH$ with $\rank(\bH)\le 2r$.

We next apply \Cref{lem:score} with $\eta=d^{-12}$. Since \(\log(d/\eta)=13\log d,\) with probability at least $1-2d^{-12}$,
\begin{equation}
\norm{\nabla\mathcal L_M(\brhostar)}_{(2r)}\le 13c_5d\sqrt r\max\left\{M^{-1}V_{\ell,K,d}\log d,M^{-1}\log d\right\}. \label{eq:ls-score-event}
\end{equation}
We now  consider the intersection of the two events above and set \(\bDelta=\widehat{\brho}_{\mathrm{LS}}-\brhostar\). Since both $\widehat{\brho}_{\mathrm{LS}}$ and $\brhostar$ belong to $\D_r$, we have  \(\rank(\bDelta)\le \rank(\widehat{\brho}_{\mathrm{LS}})+\rank(\brhostar)\le 2r\). Hence \eqref{eq:ls-rip-event} applies to $\bDelta$.

Since $\brhostar$ is feasible and $\widehat{\brho}_{\mathrm{LS}}$ is a global minimizer,  \(\mathcal L_M(\widehat{\brho}_{\mathrm{LS}})\le\mathcal L_M(\brhostar).\)  Using $\widehat{\brho}_{\mathrm{LS}}=\brhostar+\bDelta$, we therefore have
\begin{equation}
0\ge\mathcal L_M(\brhostar+\bDelta)-\mathcal L_M(\brhostar).\label{eq:ls-basic-start} \end{equation}
By the definition of $\mathcal L_M$ and $\nabla\mathcal L_M$, we have
\begin{align} \mathcal L_M(\brhostar+\bDelta)-\mathcal L_M(\brhostar)=\frac12\|\A_M(\bDelta)\|_2^2+\ip{\nabla\mathcal L_M(\brhostar)}{\bDelta}. \label{eq:ls-loss-expansion}
\end{align}
Combining \eqref{eq:ls-basic-start} and \eqref{eq:ls-loss-expansion} yields \(
\frac12\|\A_M(\bDelta)\|_2^2\le\left| \ip{\nabla\mathcal L_M(\brhostar)}{\bDelta}\right|. \)
Since $\rank(\bDelta)\le 2r$, the restricted isometry bound \eqref{eq:ls-rip-event} yields
\begin{equation}
\frac14\norm{\bDelta}_{\F}^2\le\left|\ip{\nabla\mathcal L_M(\brhostar)}{\bDelta}\right|.\label{eq:ls-basic-rip}
\end{equation}

It remains to control the inner product on the right-hand side. Since $\rank(\bDelta)\le2r$, von Neumann's trace inequality \citep{mirsky1975trace}
followed by Cauchy--Schwarz gives
\begin{align*}
\left|\ip{\bZ}{\bDelta}\right|&\le\sum_{j=1}^{\min\{2r,d\}}\sigma_j(\bZ)\sigma_j(\bDelta)\le\Bigg(\sum_{j=1}^{\min\{2r,d\}}\sigma_j^2(\bZ)\Bigg)^{1/2}\Bigg(\sum_{j=1}^{\min\{2r,d\}}\sigma_j^2(\bDelta)\Bigg)^{1/2}\cr
&=\norm{\bZ}_{(2r)}\norm{\bDelta}_{\F}.
\end{align*}
Taking $\bZ=\nabla\mathcal L_M(\brhostar)$ in \eqref{eq:ls-basic-rip}, we obtain $\frac14\norm{\bDelta}_{\F}^2\le\norm{\nabla\mathcal L_M(\brhostar)}_{(2r)}\norm{\bDelta}_{\F}$.  If $\bDelta=\bm0$, the result is immediate. Otherwise, dividing by $\norm{\bDelta}_{\F}$ gives $\norm{\bDelta}_{\F}\le 4\norm{\nabla\mathcal L_M(\brhostar)}_{(2r)}$. 
Substituting \eqref{eq:ls-score-event} therefore yields
\[\norm{\widehat{\brho}_{\mathrm{LS}}-\brhostar}_{\F}\le 52c_5d\sqrt r\max\left\{\sqrt{M^{-1}V_{\ell,K,d}\log d},M^{-1}\log d\right\}.
\]
Thus \eqref{eq:ls-bound} follows after enlarging $C_2$, if necessary, so that $C_2\ge52c_5$.

Finally, by the union bound, the probability that either the restricted isometry or the score event fails is at most \(d^{-11}+2d^{-12}\le d^{-10},\) which completes the proof.
\end{proof}

\subsection{Proof of \Cref{lem:density-projection}}\label{sec:proof_lmm41}
\begin{proof}[Proof of \Cref{lem:density-projection}] We first justify the spectral form \eqref{eq:density-projection-spectral}.
Let \[\bW=\sum_{j=1}^d\lambda_j(\bW)\bu_j\bu_j^\dagger, \qquad \lambda_1(\bW)\ge\cdots\ge\lambda_d(\bW).
\]
For any $\bX\in\D_r$, let $z_1\ge\cdots\ge z_r\ge0$ denote its
nonzero eigenvalues padded with zeros, so that $\sum_{j=1}^rz_j=1$.
By von Neumann's trace inequality,
$\Tr(\bW\bX)\le\sum_{j=1}^r\lambda_j(\bW)z_j$, with equality when
$\bX=\sum_{j=1}^rz_j\bu_j\bu_j^\dagger$. Hence
\[ \min_{\bX\in\D_r}\norm{\bW-\bX}_{\F}^2=\norm{\bW}_{\F}^2-\sum_{j=1}^r\lambda_j^2(\bW)+\min_{\substack{\bz\in\mathbb R_+^r\\ \bone^\top\bz=1}} \norm{\bz-(\lambda_1(\bW),\ldots,\lambda_r(\bW))^\top}_2^2,\]
where the ordering constraint on $\bz$ may be omitted because the Euclidean projection of the ordered vector $(\lambda_1(\bW),\ldots,\lambda_r(\bW))$ onto the simplex is itself ordered. This proves \eqref{eq:density-projection-spectral}. By optimality of the metric projection and feasibility of $\brho\in\D_r$, we have  $\norm{\Pi_{\D_r}(\bW)-\bW}_{\F}\le\norm{\brho-\bW}_{\F}$.   The triangle inequality therefore gives
\[\norm{\Pi_{\D_r}(\bW)-\brho}_{\F}\le\norm{\Pi_{\D_r}(\bW)-\bW}_{\F}+\norm{\bW-\brho}_{\F}\le2\norm{\bW-\brho}_{\F},\]
which proves \eqref{eq:density-projection-stability}. Now suppose that $\brho$ has rank $r$ and set $\varepsilon=\norm{\bW-\brho}_{\op}\le\lambda_r(\brho)/4$. Weyl's inequality gives
\(\lambda_r(\bW)\ge\lambda_r(\brho)-\varepsilon, \lambda_{r+1}(\bW)\le\varepsilon, \) 
so in particular $\lambda_r(\bW)>\lambda_{r+1}(\bW)$ and the leading $r$-dimensional eigenspace of $\bW$ is separated from the remainder. For $j\le r$, write $\lambda_j(\bW)=\lambda_j(\brho)+e_j$, where $|e_j|\le\varepsilon$, and define the candidate simplex threshold $\tau_0:=\frac{\sum_{j=1}^r\lambda_j(\bW)-1}{r}=\frac1r\sum_{j=1}^re_j$. Since $|\tau_0|\le\varepsilon$, we have, for every $j\le r$, \(\lambda_j(\bW)-\tau_0\ge\lambda_r(\bW)-\tau_0\ge\lambda_r(\brho)-2\varepsilon\ge\frac12\lambda_r(\brho)>0.\) Define $\mu_j=\lambda_j(\bW)-\tau_0$. Then $\mu_j>0$ for all $j\le r$ and, by the definition of $\tau_0$, $\sum_{j=1}^r\mu_j=\sum_{j=1}^r\lambda_j(\bW)-r\tau_0=1$.  Thus $\bmu=(\mu_1,\ldots,\mu_r)^\top$ lies in the interior of the probability simplex. Moreover, \[\mu_j-\lambda_j(\bW)+\tau_0=0,\qquad j=1,\ldots,r,\]
so $\bmu$, together with Lagrange multiplier $\tau_0$ and zero multipliers for the inactive nonnegativity constraints, satisfies the KKT conditions for
\[\min_{\substack{\bz\in\mathbb R_+^r, ~\bone^\top\bz=1}}\frac12 \norm{\bz-(\lambda_1(\bW),\ldots,\lambda_r(\bW))^\top}_2^2.
\]
The problem is strictly convex with a unique minimizer. Therefore, $\bmu$ is precisely the simplex projection appearing in \eqref{eq:density-projection-spectral}. Since every $\mu_j$ is strictly positive, $\Pi_{\D_r}\bW$ has exactly $r$ positive eigenvalues and hence $\rank \left(\Pi_{\D_r}\bW\right)=r$. This finishes the proof.
\end{proof}

\subsection{Proof of \Cref{thm:init}}\label{proofofthm:init}
\begin{proof}[Proof of \Cref{thm:init}]
We first decompose the backprojection error as
\begin{align}
\widetilde{\brho}_0-\brhostar=\frac1M\sum_{i=1}^M\left[d\Tr(\bP_i\brhostar)\bP_i-\brhostar\right]+\frac dM\sum_{i=1}^M\left[\widetilde Y_i-\Tr(\bP_i\brhostar)\right]\bP_i.\label{eq:init-error-decomposition-proof}
\end{align}
First, applying \eqref{eq:score-op} with $\eta=d^{-12}$ to the second term in \eqref{eq:init-error-decomposition-proof} gives
\begin{equation*}
\left\|\frac dM\sum_{i=1}^M\left[\widetilde Y_i-\Tr(\bP_i\brhostar)\right]\bP_i\right\|_{\op}\le 13c_5 d\left[\sqrt{\frac{V_{\ell,K,d}\log d}{M}}+\frac{\log d}{M}\right].
\end{equation*}
with probability at least $1-2d^{-12}$. For the first term in \eqref{eq:init-error-decomposition-proof}, Pauli isotropy gives zero mean, while Pauli orthogonality gives $\Ee_{\bP}\Tr(\bP\brhostar)^2=\frac{\Tr(\brhostar^2)}d$. A direct calculation then bounds the matrix variance by $d\Tr(\brhostar^2)$ and the operator norm of each summand by $2d$. Therefore, matrix Bernstein's inequality gives an absolute constant $c_6>0$ such that, with probability at least $1-2d^{-12}$,
\begin{equation*}
\left\|\frac1M\sum_{i=1}^M\left[d\Tr(\bP_i\brhostar)\bP_i-\brhostar\right]\right\|_{\op}\le c_6\left[\sqrt{\frac{d\Tr(\brhostar^2)\log d}{M}}+\frac{d\log d}{M}\right].
\end{equation*}
Combining the above estimates with \eqref{eq:init-error-decomposition-proof}, we obtain \eqref{eq:init-op} whenever $C_4\ge13c_5+c_6$. The union bound then gives probability at least $1-d^{-10}$. For the spectral truncation, Weyl's inequality gives $\lambda_r(\widetilde{\brho}_0)\ge \lambda_r(\brhostar)-\norm{\widetilde{\brho}_0-\brhostar}_{\op}\ge \frac34\lambda_r(\brhostar)>0$. Moreover, since $\brhostar$ is positive semidefinite with rank $r$, Weyl's inequality also yields
\[
\abs{\lambda_j(\widetilde{\brho}_0)}\le\norm{\widetilde{\brho}_0-\brhostar}_{\op}, \qquad j\ge r+1.
\]
It follows that $\lambda_r(\widetilde{\brho}_0)>\max_{j\ge r+1}\abs{\lambda_j(\widetilde{\brho}_0)}$. Therefore, the $r$ eigenvalues of $\widetilde{\brho}_0$ with largest magnitude are positive, and hence $\brho_0=\mathcal H_r(\widetilde{\brho}_0)$ is positive semidefinite with rank $r$. Since $\brhostar$ has rank $r$, the Eckart--Young--Mirsky theorem gives $\norm{\brho_0-\widetilde{\brho}_0}_{\op}
\le \norm{\brhostar-\widetilde{\brho}_0}_{\op}$. Therefore, by the triangle inequality,
\[\norm{\brho_0-\brhostar}_{\op}\le\norm{\brho_0-\widetilde{\brho}_0}_{\op}+\norm{\widetilde{\brho}_0-\brhostar}_{\op}\le 2\norm{\widetilde{\brho}_0-\brhostar}_{\op}.\]
Finally, since $\rank(\brho_0-\brhostar)\le 2r$, we obtain
\[
\norm{\brho_0-\brhostar}_{\F}\le\sqrt{2r}\,\norm{\brho_0-\brhostar}_{\op}\le 2\sqrt{2r}\,\norm{\widetilde{\brho}_0-\brhostar}_{\op}.\]
This completes the proof.
\end{proof}

\subsection{Deterministic contraction of QuantRGD}\label{sec:proof-local}
We establish the deterministic contraction for \Cref{alg:rgd}, which is a cornerstone in the proof of \Cref{thm:global recovery}.
\begin{lemma}\cite[Lemma~4.1]{wei2016guarantees}\label{lem:normal}
Let $\bDelta=\bX-\brhostar$, where $\bX\in\mathcal M_r^{\rm H}$. Then
\[\norm{(\Ical-\mathcal P_{T_{\bX}})\bDelta}_{\F}\le\frac{1}{\lambda_r(\brhostar)}\norm{\bDelta}_{\F}^2.\]
\end{lemma}

\begin{lemma}\label{lem:local}
Suppose that $\brhostar\in\D_r$ has rank $r$, $\brho_0\in\mathcal M_r^+$, and
\begin{equation}\label{eq:local-assumptions}
\delta_{4r}\le\frac1{64},
\qquad
\norm{\brho_0-\brhostar}_{\F}\le\frac1{16}\lambda_r(\brhostar),
\qquad
\norm{\nabla\mathcal L_M(\brhostar)}_{(2r)}\le\frac1{64}\lambda_r(\brhostar).
\end{equation}
Then $\brho_k$ is a rank-$r$ density matrix for every $k\ge1$, and
\begin{equation}\label{eq:local}\norm{\brho_k-\brhostar}_{\F}\le 4^{-k}\norm{\brho_0-\brhostar}_{\F}+4\norm{\nabla\mathcal L_M(\brhostar)}_{(2r)}.
\end{equation}
\end{lemma}

\begin{proof}
Let $\bDelta_k=\brho_k-\brhostar$, $e_k=\norm{\bDelta_k}_{\F}$, and $T_k=T_{\brho_k}\mathcal M_r^{\rm H}$. We prove by induction that $e_k\le\lambda_r(\brhostar)/16$. Set $\bDelta_k^\top=\mathcal P_{T_k}\bDelta_k$ and $\bDelta_k^\perp=(\Ical-\mathcal P_{T_k})\bDelta_k$. Since every matrix in $T_k$ has rank at most $2r$, while $\bDelta_k^\perp = -(\bI-\bU_k\bU_k^\dagger)\brhostar(\bI-\bU_k\bU_k^\dagger)$ has rank at most $r$, \Cref{lem:normal} gives
\begin{equation}\label{eq:local-normal-proof}
\norm{\bDelta_k^\perp}_{\F}\le\frac{e_k^2}{\lambda_r(\brhostar)}.
\end{equation}
Moreover, the rank-$4r$ RIP and polarization imply
\begin{equation}\label{eq:local-rip-proof}
\begin{split}
\norm{(\mathcal P_{T_k}\A_M^*\A_M\mathcal P_{T_k}-\mathcal P_{T_k})\bH}_{\F}&\le\delta_{4r}\norm{\bH}_{\F},\quad \bH\in T_k, \cr
\norm{\mathcal P_{T_k}\A_M^*\A_M\bDelta_k^\perp}_{\F}&\le\delta_{4r}\norm{\bDelta_k^\perp}_{\F},
\end{split}
\end{equation}
where the first bound follows because $\mathcal P_{T_k}\A_M^*\A_M\mathcal P_{T_k}-\mathcal P_{T_k}$ is self-adjoint on $T_k$ and the RIP controls its quadratic form; the second follows by polarizing the RIP between a rank-$2r$ element of $T_k$ and the rank-$r$ matrix $\bDelta_k^\perp$.

If $\bG_k=\bm0$, then \eqref{eq:local-rip-proof} gives $(1-\delta_{4r})\norm{\bDelta_k^\top}_{\F} \le\delta_{4r}\norm{\bDelta_k^\perp}_{\F}+\norm{\nabla\mathcal L_M(\brhostar)}_{(2r)}$.  Using \eqref{eq:local-normal-proof}, $e_k\le\lambda_r(\brhostar)/16$, and $\delta_{4r}\le1/64$, we obtain $e_k\le(64/59)\norm{\nabla\mathcal L_M(\brhostar)}_{(2r)}$; hence the claimed statistical bound already holds. We therefore assume $\bG_k\ne\bm0$ below. Since $\bG_k\in T_k$ has rank at most $2r$, exact line search and RIP give
\begin{equation}\label{eq:local-alpha-proof}
\frac1{1+\delta_{4r}}\le\alpha_k\le\frac1{1-\delta_{4r}}.
\end{equation}
The quadratic structure of $\mathcal L_M$ gives $\nabla\mathcal L_M(\brho_k)=\A_M^*\A_M\bDelta_k+\nabla\mathcal L_M(\brhostar)$. Hence, for $\bW_k=\brho_k-\alpha_k\bG_k$, we have
\begin{small}
\begin{align*}
\bW_k-\brhostar=(\Ical-\alpha_k\mathcal P_{T_k}\A_M^*\A_M\mathcal P_{T_k})\bDelta_k^\top+\bDelta_k^\perp-\alpha_k\mathcal P_{T_k}\A_M^*\A_M\bDelta_k^\perp-\alpha_k\mathcal P_{T_k}\nabla\mathcal L_M(\brhostar).
\end{align*}\end{small}%
By \eqref{eq:local-rip-proof}--\eqref{eq:local-alpha-proof}, we have
\begin{align*}
\norm{(\Ical-\alpha_k\mathcal P_{T_k}\A_M^*\A_M\mathcal P_{T_k})\bDelta_k^\top}_{\F}&\le\frac{2\delta_{4r}}{1-\delta_{4r}}\norm{\bDelta_k^\top}_{\F},\cr
\norm{\bDelta_k^\perp-\alpha_k\mathcal P_{T_k}\A_M^*\A_M\bDelta_k^\perp}_{\F}&\le\frac1{1-\delta_{4r}}\norm{\bDelta_k^\perp}_{\F}.
\end{align*}
Finally, $\mathcal P_{T_k}\nabla\mathcal L_M(\brhostar)$ has rank at most $2r$, and therefore $\norm{\mathcal P_{T_k}\nabla\mathcal L_M(\brhostar)}_{\F}\le\norm{\nabla\mathcal L_M(\brhostar)}_{(2r)}$. Combining these estimates with \eqref{eq:local-normal-proof} yields
\begin{equation}\label{eq:local-trial-proof}
\norm{\bW_k-\brhostar}_{\F}\le\frac{2\delta_{4r}}{1-\delta_{4r}}e_k+\frac1{1-\delta_{4r}}\frac{e_k^2}{\lambda_r(\brhostar)}+\frac1{1-\delta_{4r}}\norm{\nabla\mathcal L_M(\brhostar)}_{(2r)}.
\end{equation}
Using $\delta_{4r}\le1/64$ and $e_k\le\lambda_r(\brhostar)/16$, we further obtain
\begin{equation}\label{eq:local-trial-simple-proof}
\norm{\bW_k-\brhostar}_{\F}\le\frac2{21}e_k+\frac{64}{63}\norm{\nabla\mathcal L_M(\brhostar)}_{(2r)}.
\end{equation}
Together with \eqref{eq:local-assumptions}, this gives
\[\norm{\bW_k-\brhostar}_{\op}\le\norm{\bW_k-\brhostar}_{\F}\le\left(\frac1{168}+\frac1{63}\right)\lambda_r(\brhostar)<\frac14\lambda_r(\brhostar).\]
\Cref{lem:density-projection} therefore guarantees that $\brho_{k+1}=\Pi_{\D_r}(\bW_k)$ has rank $r$. Moreover, \eqref{eq:density-projection-stability} and \eqref{eq:local-trial-simple-proof} imply
\[
e_{k+1}\le\frac4{21}e_k+\frac{128}{63}\norm{\nabla\mathcal L_M(\brhostar)}_{(2r)}\le\frac14e_k+3\norm{\nabla\mathcal L_M(\brhostar)}_{(2r)}.\]
Hence, by \eqref{eq:local-assumptions}, we have
\[e_{k+1}\le\frac{\lambda_r(\brhostar)}{64}+\frac{3\lambda_r(\brhostar)}{64}=\frac{\lambda_r(\brhostar)}{16},\]
so the basin is invariant, and every projected iterate has rank $r$. Iterating the last recursion gives
\[e_k \le 4^{-k}e_0+3\norm{\nabla\mathcal L_M(\brhostar)}_{(2r)}\sum_{j=0}^{k-1}4^{-j}\le  4^{-k}e_0+4\norm{\nabla\mathcal L_M(\brhostar)}_{(2r)},\]
which proves \eqref{eq:local}.
\end{proof}

\subsection{Proof of \Cref{thm:global recovery}}\label{proofofglobalrecovery}
\begin{proof}[Proof of \Cref{thm:global recovery}]
  We first verify the restricted isometry condition. Apply \Cref{lem:pauli-rip} with $s=\min\{4r,d\}$, $\delta=1/64$, and $\eta=d^{-12}$. Since $s\le4r$, $\log(2d)\le2\log d$, and $\log(2d^{12})\le13\log d$, condition \eqref{eq:rgd-M-condition} implies \eqref{eq:pauli-rip-sample-proof} whenever $C_5$ is sufficiently large. Consequently, $\delta_{4r}\le\frac1{64}$ holds with probability at least $1-d^{-12}$.

We next verify the initialization condition. Substituting $V_{\ell,K,d}\le\frac2\ell$ and $N=M\ell$ into \eqref{eq:init-op} yields
\begin{equation*}
\norm{\widetilde{\brho}_0-\brhostar}_{\op}\le C_4\left[\sqrt{M^{-1}d\Tr(\brhostar^2)\log d}+\sqrt{2} d\sqrt{N^{-1}\log d}+M^{-1}d\log d\right].
\end{equation*}
For $C_5$ sufficiently large so that $C_4\left(\frac{1+\sqrt2}{\sqrt{C_5}}+\frac1{C_5}\right)\le \frac1{32\sqrt2}$, substituting the sample-size conditions \eqref{eq:rgd-M-condition}, \eqref{eq:rgd-N-condition} into the initialization bound above, we obtain
$\norm{\widetilde{\brho}_0-\brhostar}_{\op}\le\frac{\lambda_r(\brhostar)}{32\sqrt{2r}}$.
Together with \eqref{eq:init-frobenius-bound}, it gives
\begin{equation}\label{eq:end-frobenius-basin-proof}
\norm{\brho_0-\brhostar}_{\F}\le \frac1{16}\lambda_r(\brhostar).
\end{equation}

Finally, we verify the score condition. Applying \eqref{eq:score-restricted} with $\eta=d^{-12}$ and substituting $V_{\ell,K,d}\le\frac2\ell$ gives, with probability at least $1-2d^{-12}$, it holds
\begin{equation}\label{eq:end-score-bound-proof}
\norm{\nabla\mathcal L_M(\brhostar)}_{(2r)}\le13c_5d\sqrt r\left[\sqrt{2N^{-1}\log d}+M^{-1}\log d\right].
\end{equation}
For $C_5$ sufficiently large so that $13c_5\left(\sqrt{\frac2{C_5}}+\frac1{C_5}\right)\le \frac1{64}$, substituting \eqref{eq:rgd-M-condition}, \eqref{eq:rgd-N-condition} into \eqref{eq:end-score-bound-proof}, we obtain $\norm{\nabla\mathcal L_M(\brhostar)}_{(2r)}\le\frac1{64}\lambda_r(\brhostar)$.
Thus, applying \Cref{lem:local} and substituting \eqref{eq:end-score-bound-proof} into \eqref{eq:local}, we obtain \eqref{eq:global recovery} whenever $C_5\ge52\sqrt2\,c_5$. Finally, the union bound over the restricted isometry, initialization, and score events gives failure probability at most $2d^{-10}$.  
\end{proof}

\section{Conclusion}\label{sec:conclusion}
We studied low-rank quantum state tomography from finite-bit Pauli batch responses. We introduced HyperQuant, a mean-preserving hyperbolic quantizer that keeps the finite-bit responses conditionally unbiased while controlling quantization variance, and established minimax distortion guarantees. Exact mean preservation enables direct rank-constrained least-squares recovery, for which we derived nonasymptotic guarantees and a favorable bit--shot tradeoff relative to unquantized batch transmission. We further developed QuantRGD, a Riemannian gradient method with provable linear convergence to the statistical neighborhood under explicit resource conditions. Numerical experiments validate the quantization, recovery, and convergence performance predicted by the theory. 
{\small
\bibliographystyle{plain}
\bibliography{qst}

@article{wei2016guarantees,
  title={Guarantees of Riemannian optimization for low rank matrix recovery},
  author={Wei, Ke and Cai, Jian-Feng and Chan, Tony F. and Leung, Shingyu},
  journal={SIAM Journal on Matrix Analysis and Applications},
  volume={37},
  number={3},
  pages={1198--1222},
  year={2016},
}

@article{zyczkowski2001induced,
  title={Induced measures in the space of mixed quantum states},
  author={{\.Z}yczkowski, Karol and Sommers, Hans-J{\"u}rgen},
  journal={Journal of Physics A: Mathematical and General},
  volume={34},
  number={35},
  pages={7111--7125},
  year={2001}
}

@article{mirsky1975trace,
  title={A trace inequality of John von Neumann},
  author={Mirsky, Leon},
  journal={Monatshefte f{\"u}r Mathematik},
  volume={79},
  number={4},
  pages={303--306},
  year={1975},
  publisher={Springer}
}

@article{mezzadri2007,
  author  = {Francesco Mezzadri},
  title   = {How to Generate Random Matrices from the Classical Compact Groups},
  journal = {Notices of the American Mathematical Society},
  volume  = {54},
  number = {5},
  pages   = {592--604},
  year    = {2007}
}

@article{cai2010singular,
  title   = {A Singular Value Thresholding Algorithm for Matrix Completion},
  author  = {Cai, Jian-Feng and Cand{\`e}s, Emmanuel J. and Shen, Zuowei},
  journal = {SIAM Journal on Optimization},
  volume  = {20},
  number  = {4},
  pages   = {1956--1982},
  year    = {2010},
}

@article{reilly2015interface,
  author  = {Reilly, David J.},
  title   = {Engineering the quantum-classical interface of solid-state qubits},
  journal = {npj Quantum Information},
  volume  = {1},
  pages   = {15011},
  year    = {2015},
}

@article{danjou2021readout,
  author  = {D'Anjou, B.},
  title   = {Generalized figure of merit for qubit readout},
  journal = {Physical Review A},
  volume  = {103},
  number  = {4},
  pages   = {042404},
  year    = {2021},
}

@article{google2025qec,
  author  = {{Google Quantum AI and Collaborators}},
  title   = {Quantum error correction below the surface code threshold},
  journal = {Nature},
  volume  = {638},
  pages   = {920--926},
  year    = {2025},
}

@inproceedings{prathapan2022cryo,
  author    = {Prathapan, Mridula and Mueller, Peter and Heim, David and Oropallo, Maria Vittoria and Br{\"a}ndli, Matthias and Francese, Pier Andrea and Kossel, Marcel and Ruffino, Andrea and Zota, Cezar and Cha, Eunjung and Morf, Thomas},
  title     = {A system design approach toward integrated cryogenic quantum control systems},
  booktitle = {2022 IEEE 15th Workshop on Low Temperature Electronics (WOLTE)},
  pages     = {1--4},
  year      = {2022},
  publisher = {IEEE},
}

@article{carreravazquez2024realtime,
  author  = {Carrera Vazquez, Almudena and Tornow, Caroline and Rist{\`e}, Diego and Woerner, Stefan and Takita, Maika and Egger, Daniel J.},
  title   = {Combining quantum processors with real-time classical communication},
  journal = {Nature},
  volume  = {636},
  pages   = {75--79},
  year    = {2024},
}

@article{lybrand2019quantization,
  author  = {Lybrand, Eric and Saab, Rayan},
  title   = {Quantization for low-rank matrix recovery},
  journal = {Information and Inference: A Journal of the IMA},
  volume  = {8},
  number  = {1},
  pages   = {161--180},
  year    = {2019},
}

@article{gross2010qst,
  author  = {Gross, David and Liu, Yi-Kai and Flammia, Steven T. and Becker, Stephen and Eisert, Jens},
  title   = {Quantum state tomography via compressed sensing},
  journal = {Physical Review Letters},
  volume  = {105},
  number  = {15},
  pages   = {150401},
  year    = {2010},
}

@inproceedings{liu2011universal,
  author    = {Liu, Yi-kai},
  title     = {Universal low-rank matrix recovery from {Pauli} measurements},
  booktitle = {Advances in Neural Information Processing Systems 24},
  pages     = {1638--1646},
  year      = {2011},
  publisher = {Curran Associates, Inc.},
}

@article{liu2012compressed,
  author  = {Flammia, Steven T. and Gross, David and Liu, Yi-Kai and Eisert, Jens},
  title   = {Quantum tomography via compressed sensing: error bounds, sample complexity and efficient estimators},
  journal = {New Journal of Physics},
  volume  = {14},
  number  = {9},
  pages   = {095022},
  year    = {2012},
}

@article{xia2017pauli,
  author  = {Xia, Dong},
  title   = {Estimation of low rank density matrices by {Pauli} measurements},
  journal = {Electronic Journal of Statistics},
  volume  = {11},
  number  = {1},
  pages   = {50--77},
  year    = {2017},
}

@article{guta2020fast,
  author  = {Gu{\c{t}}{\u{a}}, M{\u{a}}d{\u{a}}lin and Kahn, Jonas and Kueng, Richard and Tropp, Joel A.},
  title   = {Fast state tomography with optimal error bounds},
  journal = {Journal of Physics A: Mathematical and Theoretical},
  volume  = {53},
  number  = {20},
  pages   = {204001},
  year    = {2020},
}

@inproceedings{zhang2013communication,
  author    = {Zhang, Yuchen and Duchi, John C. and Jordan, Michael I. and Wainwright, Martin J.},
  title     = {Information-theoretic lower bounds for distributed statistical estimation with communication constraints},
  booktitle = {Advances in Neural Information Processing Systems 26},
  pages     = {2328--2336},
  year      = {2013},
  publisher = {Curran Associates, Inc.},
}

@inproceedings{han2018geometric,
  author    = {Han, Yanjun and {\"O}zg{\"u}r, Ayfer and Weissman, Tsachy},
  title     = {Geometric lower bounds for distributed parameter estimation under communication constraints},
  booktitle = {Proceedings of the 31st Conference on Learning Theory},
  series    = {Proceedings of Machine Learning Research},
  volume    = {75},
  pages     = {3163--3188},
  year      = {2018},
  publisher = {PMLR},
}

@article{barnes2020fisher,
  author  = {Barnes, Leighton Pate and Han, Yanjun and {\"O}zg{\"u}r, Ayfer},
  title   = {Lower bounds for learning distributions under communication constraints via {Fisher} information},
  journal = {Journal of Machine Learning Research},
  volume  = {21},
  number  = {236},
  pages   = {1--30},
  year    = {2020},
}

@inproceedings{suresh2017mean,
  author    = {Suresh, Ananda Theertha and Yu, Felix X. and Kumar, Sanjiv and McMahan, H. Brendan},
  title     = {Distributed mean estimation with limited communication},
  booktitle = {Proceedings of the 34th International Conference on Machine Learning},
  series    = {Proceedings of Machine Learning Research},
  volume    = {70},
  pages     = {3329--3337},
  year      = {2017},
  publisher = {PMLR},
}

@misc{doosti2026distributed,
  author        = {Doosti, Mina and Sweke, Ryan and Wadhwa, Chirag},
  title         = {Distributed quantum property testing with communication constraints},
  year          = {2026},
}

@misc{mattig2026distributed,
  author        = {M{\"a}ttig-V{\'a}squez, Hans and Delgado, Aldo and Pereira, Luciano},
  title         = {Rigorous quantum state tomography for distributed quantum computing},
  year          = {2026},
}

@article{kyrillidis2018provable,
  author  = {Kyrillidis, Anastasios and Kalev, Amir and Park, Dohyung and Bhojanapalli, Srinadh and Caramanis, Constantine and Sanghavi, Sujay},
  title   = {Provable compressed sensing quantum state tomography via non-convex methods},
  journal = {npj Quantum Information},
  volume  = {4},
  pages   = {36},
  year    = {2018},
}

@article{kim2023fast,
  author  = {Kim, Junhyung Lyle and Kollias, George and Kalev, Amir and Wei, Ken X. and Kyrillidis, Anastasios},
  title   = {Fast quantum state reconstruction via accelerated non-convex programming},
  journal = {Photonics},
  volume  = {10},
  number  = {2},
  pages   = {116},
  year    = {2023},
}

@article{hsu2024rgd,
  author  = {Hsu, Ming-Chien and Kuo, En-Jui and Yu, Wei-Hsuan and Cai, Jian-Feng and Hsieh, Min-Hsiu},
  title   = {Quantum state tomography via nonconvex {Riemannian} gradient descent},
  journal = {Physical Review Letters},
  volume  = {132},
  number  = {24},
  pages   = {240804},
  year    = {2024},
}

@article{cai2013maxnorm,
  author  = {Cai, T. Tony and Zhou, Wen-Xin},
  title   = {A Max-Norm Constrained Minimization Approach to 1-Bit Matrix Completion},
  journal = {Journal of Machine Learning Research},
  volume  = {14},
  number  = {114},
  pages   = {3619--3647},
  year    = {2013},
}

@article{davenport2014onebit,
  author  = {Davenport, Mark A. and Plan, Yaniv and van den Berg, Ewout and Wootters, Mary},
  title   = {1-Bit matrix completion},
  journal = {Information and Inference: A Journal of the IMA},
  volume  = {3},
  number  = {3},
  pages   = {189--223},
  year    = {2014},
}

@article{klopp2015adaptive,
  author  = {Klopp, Olga and Lafond, Jean and Moulines, {\'E}ric and Salmon, Joseph},
  title   = {Adaptive Multinomial Matrix Completion},
  journal = {Electronic Journal of Statistics},
  volume  = {9},
  number  = {2},
  pages   = {2950--2975},
  year    = {2015},
}

@article{bhaskar2016probabilistic,
  author  = {Bhaskar, Sonia A.},
  title   = {Probabilistic Low-Rank Matrix Completion from Quantized Measurements},
  journal = {Journal of Machine Learning Research},
  volume  = {17},
  number  = {60},
  pages   = {1--34},
  year    = {2016},
}

@inproceedings{ni2016optimal,
  author    = {Ni, Renkun and Gu, Quanquan},
  title     = {Optimal Statistical and Computational Rates for One Bit Matrix Completion},
  booktitle = {Proceedings of the 19th International Conference on Artificial Intelligence and Statistics},
  series    = {Proceedings of Machine Learning Research},
  volume    = {51},
  pages     = {426--434},
  year      = {2016},
  editor    = {Gretton, Arthur and Robert, Christian C.},
  publisher = {PMLR},
  address   = {Cadiz, Spain},
}

@inproceedings{shen2019robust,
  author    = {Shen, Jie and Awasthi, Pranjal and Li, Ping},
  title     = {Robust Matrix Completion from Quantized Observations},
  booktitle = {Proceedings of the Twenty-Second International Conference on Artificial Intelligence and Statistics},
  series    = {Proceedings of Machine Learning Research},
  volume    = {89},
  pages     = {397--407},
  year      = {2019},
  editor    = {Chaudhuri, Kamalika and Sugiyama, Masashi},
  publisher = {PMLR},
}

@article{chen2023dithered,
  author  = {Chen, Junren and Wang, Cheng-Long and Ng, Michael K. and Wang, Di},
  title   = {High Dimensional Statistical Estimation under Uniformly Dithered One-bit Quantization},
  journal = {IEEE Transactions on Information Theory},
  volume  = {69},
  number  = {8},
  pages   = {5151--5187},
  year    = {2023},
}

@article{huang2018robust,
  author  = {Huang, Jian and Jiao, Yuling and Lu, Xiliang and Zhu, Liping},
  title   = {Robust Decoding from 1-Bit Compressive Sampling with Ordinary and Regularized Least Squares},
  journal = {SIAM Journal on Scientific Computing},
  volume  = {40},
  number  = {4},
  pages   = {A2062--A2086},
  year    = {2018},
}

@misc{ding2020nonquadratic,
  author        = {Ding, Lijun and Zhang, Yuqian and Chen, Yudong},
  title         = {Low-Rank Matrix Recovery with Non-Quadratic Loss: Projected Gradient Method and Regularity Projection Oracle},
  year          = {2020},
}

@article{smith2026provableEDMC,
  title={Provable non-convex {Euclidean} distance matrix completion: geometry, reconstruction, and robustness},
  author={Smith, Chandler and Cai, HanQin and Tasissa, Abiy},
  journal={IEEE Transactions on Information Theory},
  year={2026},
  publisher={IEEE}
}

@article{vandereycken2013low,
  title={Low-rank matrix completion by {Riemannian} optimization},
  author={Vandereycken, Bart},
  journal={SIAM Journal on Optimization},
  volume={23},
  number={2},
  pages={1214--1236},
  year={2013},
  publisher={SIAM}
}

@inproceedings{hamm2022RieCUR,
  title = {Riemannian {CUR} decompositions for robust principal component analysis},
  author = {Hamm, Keaton and Meskini, Mohamed and Cai, HanQin},
  booktitle = {Topological, Algebraic and Geometric Learning Workshops 2022},
  pages = {152--160},
  year = {2022},
  publisher = {PMLR}
}

@article{cai2019accelerated,
  title={Accelerated alternating projections for robust principal component analysis},
  author={Cai, HanQin and Cai, Jian-Feng and Wei, Ke},
  journal={Journal of Machine Learning Research},
  volume={20},
  number={20},
  pages={1--33},
  year={2019}
}
}

\end{document}